%% file: bare_jrnl.tex
\documentclass[12pt, draftclsnofoot, onecolumn]{IEEEtran}
\input{macro_en}

\usepackage{amsthm}
\usepackage{amssymb}
\usepackage{graphicx}
\usepackage{subfigure}
\usepackage{tabularx}
\usepackage{color,cite}
\usepackage{booktabs}
\usepackage{url}
\usepackage{bm}
\usepackage{color}
\definecolor{c}{rgb}{1,0,0} % red
\definecolor{b}{rgb}{0,0,1} % red

\usepackage{algpseudocode}
\usepackage{algorithm}

\usepackage{makecell}
\usepackage{amsfonts}

\usepackage{graphicx}
\usepackage{float}
\usepackage[caption=false,font=footnotesize]{subfig}

\newtheorem{remark}{Remark}
\newtheorem{theorem}{Theorem}

\usepackage{algorithm}
\newtheorem{proposition}{Proposition}

\usepackage{algpseudocode}
\usepackage{cite}

\begin{document}
	
%	\title{ Secret Key Generation Relying on Reconfigurable Intelligent Surface}
%\title{Multicell Secret Key Generation Relying on Reconfigurable Intelligent Surface}
	\title{Reconfigurable Intelligent Surface-aided Secret Key Generation in Multi-Cell Systems}
	\author{
		Lei~Hu,
		Chen~Sun,~\IEEEmembership{Member,~IEEE},
		Guyue~Li,~\IEEEmembership{Member,~IEEE},
		 Aiqun~Hu,~\IEEEmembership{Senior Member,~IEEE}, and  Derrick Wing Kwan Ng,~\IEEEmembership{Fellow,~IEEE}
		
		\thanks{(Corresponding author: Guyue Li.)}
		\thanks{Lei Hu and Guyue Li are with the School of Cyber Science and Engineering, Southeast University, Nanjing 210096, China. Guyue Li is also with Purple Mountain Laboratories, Nanjing 211111, China, and also with the Jiangsu Provincial Key Laboratory of Computer Network Technology, Nanjing 210096, China (e-mail: lei-hu@seu.edu.cn; guyuelee@seu.edu.cn.).}
		
		\thanks{Chen Sun and Aiqun Hu are with the National Mobile Communications
			Research Laboratory, Southeast University, Nanjing 210096, China, and also
			with the Purple Mountain Laboratories, Nanjing 211100, China (e-mail:
			sunchen@seu.edu.cn; aqhu@seu.edu.cn). Aiqun Hu is also with the Jiangsu Provincial Key Laboratory of Computer Network Technology, Nanjing 210096, China.
		}

		\thanks{Derrick Wing Kwan Ng is with the School of Electrical Engineering and Telecommunications, University of New South
		Wales, Sydney, NSW 2052, Australia (e-mail: w.k.ng@unsw.edu.au).} 
	}

	\maketitle
	
	\begin{abstract}
		Physical-layer key generation (PKG) exploits the reciprocity and randomness of wireless channels to generate a symmetric key between two legitimate communication ends. However, in multi-cell systems, PKG suffers from severe pilot contamination due to the reuse of pilots in different cells.
		In this paper, we invoke multiple reconfigurable intelligent surfaces (RISs) for adaptively shaping the environment and enhancing the PKG performance.
		To this end, we formulate an optimization problem to maximize the weighted sum key rate (WSKR) by jointly optimizing the precoding matrices at the base stations (BSs) and the phase shifts at the RISs.
		For addressing the non-convexity of the problem, we derive an upper bound of the WSKR and prove its tightness.
		To tackle the upper bound maximization problem, we apply an alternating optimization (AO)-based algorithm to divide the joint optimization into two sub-problems.
		We apply the Lagrangian dual approach based on the Karush-Kuhn-Tucker (KKT) conditions for the sub-problem of  precoding matrices and adopt a projected gradient ascent (PGA) algorithm for the sub-problem of  phase shifts.
		Simulation results confirm the near-optimal performance of the proposed algorithm and the effectiveness of RISs for improving the WSKR via mitigating pilot contamination.
%		introduction of RISs can effectively improve the WSKR by alleviating the interference caused by the pilot contamination.
		%		Simulation results verify the achieved WSKR by the proposed algorithm approaches the upper bound, showing the near-optimal performance.
%		In addition, the introduction of RISs can effectively improve the WSKR by alleviating the interference caused by the pilot contamination.
%		Finally, a lower bit disagreement rate (BDR) can be achieved by deploying a large number of RIS elements.
		%		For the sub-problem of optimizing the precoding matrices at the BSs, we apply the Lagrangian dual approach based on the Karush-Kuhn-Tucker (KKT) conditions, while for the sub-problem of designing the phase shifts at the RISs, we adopt a projected gradient ascent (PGA) algorithm.
%		 confirming that the proposed scheme can effectively promote.
%		 our results unveil that in PKG systems, deploying a  RIS close to the UTs can provide higher performance gain while deploying multiple RISs is always beneficial compared to single RIS counterparts.
	\end{abstract}
	
	
	\section{Introduction}
	% 高速率需要高安全性-》高层加密的问题->keyless的问题->PKG问题；
	% 用RIS带来密钥性能增益。->现有的文章都是考虑单小区系统->多小区下干扰问题。
	% 已有的方法表明RIS能有效抑制下行干扰，但是在PKG中，我们需要同时处理上行探测和下行探测的干扰，需要对RIS反射系数进行设计。
%		Wireless communication systems provide higher data rates driven by emerging applications
		
		The  ever-increasing connectivity among a large number of devices and the ubiquitous wireless communications have aroused great awareness to establish secure communication on the fly \cite{ylianttila20206g}.
	%		However, the broadcast and superposition properties of wireless medium impose great threat to the data confidentiality. 
		Conventionally, secure communication is guaranteed by applying cryptographic encryption mechanisms in the application layer \cite{ZhangReview}. Particularly, symmetric keys should be distributed to legitimate parties in these mechanisms before the actual communication takes place. However, existing cryptographic techniques face difficulties to realize secret key sharing in ad-hoc and mobile networks \cite{JiaoWCM}.
%		However, this is challenging to realize in practical systems, such as ad-hoc and mobile networks \cite{JiaoWCM}. 
		As an alternative, physical-layer key generation (PKG) exploits the intrinsic reciprocity and randomness of wireless channels to generate a pair of secret keys between the desired legitimate ends \cite{maurer1993secret}.
		Furthermore, due to the existence of spatial decorrelation, an eavesdropper, Eve, cannot obtain any information about the generated keys if she locates more than half a wavelength away from the legitimate ends, i.e., Alice and Bob \cite{Guillaume2015Bringing,2021Sum,mathur2008radio}. 
%		Since it 
%		in a cost-effective manner, PKG has received great attention in the literature and 

	%
	The process of PKG generally consists of four steps: channel sounding, quantization, information reconciliation, and privacy amplification \cite{mathur2008radio}. During the channel probing step, Alice and Bob exchange pilots to acquire highly correlated channel estimations.
	Then, the extracted channel estimations are respectively quantized into bit sequences at Alice and Bob in the quantization step. 
%	However, due to the imperfect reciprocity in the channel probing step, the extracted bits at Alice and Bob are not identical. Hence, 
	Also, during the information reconciliation step, error-correcting codes are adopted to correct the mismatched bits between Alice and Bob.
	Finally, privacy amplification is employed to erase the bits that might have leaked information to Eve in the previous probing and reconciliation steps.
	From the above steps, it can be seen that PKG highly relies on the inherent randomness and reciprocity of wireless channels.
%	It should be noted that the channel estimations at Alice and Bob need to be highly correlated. 
	However, the desired secret key rate may not be guaranteed in some harsh propagation environments, such as wave-blockage environments \cite{Li2022WCM}, \cite{22Sum-RIS}. 
	Fortunately, reconfigurable intelligent surface (RIS), which has emerged as a disruptive wireless communication technology,
%	is capable of controlling the propagation environment,
	 has great potential to address this problem. 
%	 in a cost-effective and energy-effective way
	 In fact, RIS is a planar surface comprising a large number of low-cost passive reflecting elements \cite{Smart,di2020reconfigurable,pan2021reconfigurable}. These elements can independently adjust the phase shifts to collaboratively customize the wireless propagation environment \cite{pan2021reconfigurable}. 
	 Therefore, it is expected that deploying RIS in secure communication  can facilitate the required PKG. In particular, when the direct link between Alice and Bob is blocked, RIS could shape a RIS-induced fluctuating channel to serve as a controllable randomizer for generating a secret key.

		However, to fully unleash the potential of the RIS for improving PKG performance, the optimization of the phase shifts of RIS is required. To date, several studies have focused on the design of phase shifts of RIS \cite{2021SPL, 21JiTVT, 22Sum-RIS, Hu_MISO, Lu_MISO, Chen_MISO}. 
		For instance, \cite{2021SPL,21JiTVT,22Sum-RIS}		studied the reflection coefficients optimization  in  single-input
		single-output (SISO) systems with only one legitimate user. 
		In particular, the authors in \cite{2021SPL}  assumed the RIS-induced channel of the eavesdropper is independent	from that of the legitimate ends. Based on this, they 	derived the expression of the key generation rate (KGR) capacity and optimized the on/off states of the RIS units. Furthermore, in \cite{21JiTVT}, the authors considered RIS-assisted PKG with multiple non-colluding eavesdroppers. They designed a semidefinite relaxation
		(SDR) and successive convex approximation (SCA)-based algorithm to maximize the secret key capacity lower bound.
		Then, a RIS-assisted multiuser key generation scheme was studied in  \cite{22Sum-RIS}, where the RIS configuration was optimized to maximize the sum secret key rate of multiple users. On the other hand, \cite{Hu_MISO,Lu_MISO,Chen_MISO} 	investigated the beamforming optimization in multiple-input single-output (MISO) systems. In particular, \cite{Hu_MISO} proposed a low-complexity block successive upper-bound minimization (BSUM) with the mirror-prox method to optimize the reflective beamforming at the RIS and the transmit beamforming at the	base station (BS), with the consideration of the spatial correlation at both the BS and the RIS. In addition, 
%		to optimize the precoding matrix at the BS and phase-shift matrix at the RIS,
		 \cite{Lu_MISO} treated the coupled precoding matrix and phase-shift matrix as an equivalent variable and designed a water-filling algorithm to acquire its optimal solution. Then, they recovered the two matrices from the optimized variable. 
		Furthermore, to  obtain a computationally efficient suboptimal solution to the non-convex problem in \cite{Lu_MISO}, the authors in \cite{Chen_MISO} adopted a machine learning-based algorithm to achieve a higher KGR.
	Nevertheless, all of these works i.e., \cite{2021SPL, 21JiTVT, 22Sum-RIS, Hu_MISO, Lu_MISO, Chen_MISO}, focus on the design of single-cell systems, while practical RIS-based PKG methods in multi-cell systems are still lacked. Indeed, in multi-cell systems, the same pilot pool is reused in different cells due to the limited time and frequency resources giving rise to the pilot contamination problem \cite{Jose,yinhaifan,Jose_TWC}. 
	Over the last couple of years, considerable research efforts have been devoted for studying the impact of pilot contamination on spectral and energy efficiencies and the corresponding methods for alleviating the negative impacts caused by pilot interference \cite{yinhaifan,Jose_TWC,6288608,6756975,7236930}. However, all of these works aim for reducing the uplink channel estimation error that do not align the goal of PKG which aims for  improving the reciprocity between channel estimations in the uplink and downlink. Indeed,  pilot contamination exists in both the uplink and downlink channel probing phase in PKG systems that introduces non-reciprocal interference to the channel estimations at the BSs and the user terminals (UTs).
%	Specifically, the interference in the uplink and downlink arise from the UTs and the BSs in other cells, respectively. 
	Therefore, this channel asymmetry caused by multi-cell pilot contamination is expected to jeopardize the key generation performance. 
%	In addition, only a single RIS is considered in these methods
	To tackle this problem, we propose to employ multiple RISs to assist the PKG in multi-cell networks. 
	Specifically, by carefully altering the phase shifts introduced by multiple RISs, the inter-cell pilot interference reflected by the RISs and the interference in the direct channel can be harnessed such that they can be destructively superimposed at the desired communication nodes to minimize the interference power. 
	Therefore, the proposed paradigm provides a new degrees of freedom (DoF) to facilitate multi-cell PKG in conjunction with the precoding matrices at the BSs. 
	However, the BSs' precoding matrices and the RISs' phase shifts have to be jointly optimized to fully unleash the potential of RISs for effective KGR provisioning. More importantly, existing techniques in \cite{2021SPL, 21JiTVT, 22Sum-RIS, Hu_MISO, Lu_MISO, Chen_MISO} cannot be directly applied to multi-cell PKG systems since the inter-cell pilot contamination is not taken into consideration.

%	However, it is not straightforward to extend these works to multi-cell scenarios, \blue{ since the PKG model under multi-cell pilot contamination is more complicated and the optimization objective is different.} 
	
	To address the above issues, this paper investigates the PKG method in multi-cell systems. We introduce RISs to combat multi-cell pilot contamination by jointly optimizing the precoding matrices at the BSs and the phase shifts at the RISs.
%	 To address the above issues, this paper investigates the RIS-aided PKG method in multicell systems.
	 %We derive the KGR expression and present an effective algorithm to jointly optimize the transmit beamforming and reflective beamforming.
%	 To address pilot contamination problem in PKG systems, this paper studies the RIS-aided PKG methos in multicell systems. 
%	 In particular, this paper employ the RIS to tackle the challenging problem of PKG under multicell pilot contamination.
	 More specifically, the main contributions of this paper are as follows.
	 \begin{itemize}
%	 	\item To the best of the authors' knowledge, this is the first work to show the impact of multicell pilot contamination on PKG performance.
%		\item 
		%		To the best of our knowledge, this is the first work to investigate PKG in multi-cell systems.	
%		We present the traditional secret key generation scheme and prove that the KGR is reduced significantly under multi-cell pilot contamination. To tackle this problem, we adopt RISs to mitigate the inter-cell interference. 
%	 	Specifically, 
	 	\item We propose a novel RIS-aided multi-cell PKG framework based on the precoding matrices at the BSs and the phase shifts at the RISs. We then derive a closed-form KGR expression that facilitates the formulation of an optimization problem to maximize the weighted sum key rate (WSKR) of all the cells. Since the formulated problem is non-convex and	difficult to solve, we derive a tight upper bound of the WSKR and maximize the upper bound.
%	 	\blue{}.
	 	
%	 	We formulate the design of beamforming as an optimization problem to maximize the minimum KGR for the	 	worst-case eavesdropper channel.	
%	 	Furthermore, our analysis shows that
	 	%	we analyze the KGR performance difference
	 	%		between the one adopting the assumption of the i.i.d. model
	 	%		and that of the spatially correlated model. 
	 	%	It is found that 
%	 	the
%	 	beamforming designed for the correlated model outperforms that
%	 	for the i.i.d. model while the KGR gain increases with the channel
%	 	correlation with the proposed design.

	 	\item 
	 	To tackle the upper bound maximization problem, we employ an alternating optimization (AO)-based algorithm to alternately obtain a high-quality suboptimal solution. To be specific, a Lagrangian dual algorithm based on the KKT conditions is applied to design the precoding matrices at the BSs and a PGA
	 	algorithm is adopted to optimize the phase shifts at the RISs. 
	 	
%	 	To tackle the resulting non-convex optimization problem, we present an effective Block Successive Upper-bound Minimization
%	 	(BSUM)-based algorithm. 
%	 	We prove that the BSUM algorithm yields a non-decreasing convergence over iterations.
	 	%	 The required number of iterations to convergence is small in simulations, which reduces the computation cost. 
%	 	Then, to solve the non-smooth convex problem in each iteration of the BSUM algorithm in a complexity-effective manner, we reformulate it as an equivalent convex-concave saddle point problem and employ the Mirror-Prox method to solve it with closed-form updates. 
	 	%	In each iteration of the proposed method, the solution is in closed form, and thus the computation complexity is low.
	 	
	 	\item Simulation results verify that the WSKR by the proposed algorithm approaches the upper bound, showing the near-optimal performance.
%	 	without RISs, the WSKR of cell-edge users approaches zero, while the proposed RIS-aided method can achieve high key rate. 
		Also, a significant WSKR  improvement can be observed with the increase of RIS elements number. 
	 	Finally, distributed RISs, with each RIS being located in the proximity of some  UTs, offer rich spatial diversity to maximize the WSKR.
%	 	Finally, RIS deployed in the proximity of UTs achieves the maximum WSKR and   deployment is beneficial for enhancing the	key generation performance.
	 	
%	 	The proposed RIS-assisted scheme can achieve high KGR and the performance gain increase with the number of RIS elements. Also, the introduce of RIS reduce the BDR between BSs and UTs, which 
%	 	Finally, a favorable BS-RIS link and RIS-UT link could increase the key generation performance.
%	 	
%	 	
%	 	\blue{Simulation results show that 
%	 	compared to existing
%	 	methods based on the i.i.d. fading model, the proposed
%	 	design achieves about $5$ dB transmit power  gain when the spacing between two neighbouring RIS elements is a quarter of the wavelength and the BS antenna correlation $\rho$ is 0.2. Also, the KGR gain increases with
%	 	the spatial correlation at both the BS and the RIS.
%	 	Moreover, the computational time of the proposed
%	 	algorithm is reduced approximately by two orders of magnitude compared to that of the
%	 	commonly adopted algorithms, e.g., alternating optimization, semidefinite relaxation, successive convex approximation with Gaussian randomization (ASSG), while achieving a similar KGR performance. }
	 \end{itemize}

	\emph{Notations:} 
	In this paper, $\mathbb{C}^{A\times B}$ denotes the space of complex matrices of size $A\times B$. Matrices and vectors are
	denoted by boldface capital and lower-case letters, respectively.
%	$\Re(\cdot)$ and $\Im(\cdot)$ stand for the real and imaginary parts of a complex number. 
	The imaginary unit of a complex number is denoted by $j=\sqrt{-1}$.
	$(\cdot)^T$ and  $(\cdot)^H$ denote  transpose and conjugate transpose, respectively. $\diag(\x)$
	denotes a diagonal matrix whose diagonal elements
	are extracted from vector $\x$.
	$\rm{vec}(\X)$ denotes the vectorization of matrix $\X$.
%	$\tr(\cdot)$ represents the trace of a matrix.
	 $\tr(\cdot)$  represents the trace of a matrix. 
%	 and $\text{rank}(\cdot)$
	 $\triangleq$ means “defined as”.
	$\otimes$ denotes the  Kronecker product.
%	The Kronecker product, Hadamard product, and Khatri-Rao product are represented by $\otimes$, $\circ$, and $\odot$, respectively. 
	$\mathcal{I}(X;Y)$ 
%	is the mutual information of random variables $X$ and $Y$.
	and $\mathcal{H}(X,Y)$ 
	are the mutual information and joint entropy of random variables $X$ and $Y$, respectively. 
%	$\mathcal{H}(X,Y|Z)$ is the conditional entropy of $X$ and $Y$ given $Z$.
	$\operatorname{det}(\cdot)$ is the matrix determinant. 
%	$||\boldsymbol{x}||_{1}$, $||\boldsymbol{x}||_{2}$, and $||\boldsymbol{x}||_{\infty}$ denote the $\ell_{1}$, $\ell_{2}$, and $\ell_{\infty}$ norms of vector $\boldsymbol{x}$. 
	$\|\X\|_{F}$ is the Frobenius norm of matrix $\X$.
	$\mathbb{E}\{\cdot\}$ represents statistical expectation. 
	$\lambda_{\ell}(\X)$ is the $\ell$-th largest eigenvalue of matrix $\X$.
	$\mathcal{O}(\cdot)$ is the big-O notation.
%	$\X \succeq \boldsymbol{0}$ means $\X$ is a positive semidefinite matrix.
	$\left[\X\right]_{m:n,i}$ means the matrix consisting of rows $m$ to $n$ and column $i$ of matrix $\X$.
	$\nabla f(\cdot)$ 
	and 
	$\partial f/\partial x$ are 
%	is 
	the gradient operator of function $f$. $\I_M$ denotes the identity matrix  of dimension $M$. 

	\section{RIS-based Multi-Cell Key Generation Model} \label{sec:framework}
	%	\blue{The introduce of RISs }
	%	In this section, we first introduce the RIS-based channel model. Then, we present a RIS-aided PKG framework and give the specific channel probing process based on the	precoding matrices at the BSs and the phase-shifting vector at the RISs.
%	\subsection{System Model}
		\begin{figure}
		\centering
		\includegraphics[width=0.53\textwidth]{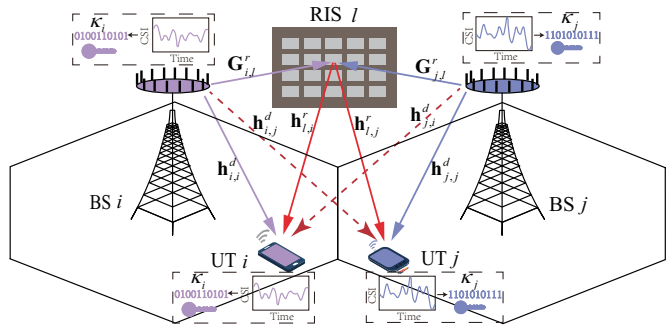}
		\caption{The model of RIS-aided PKG in multi-cell systems.}
		\label{fig:system_model}
	\end{figure}
	As shown in Fig. \ref{fig:system_model}, we consider a multi-cell PKG model constituted by $K$ cells, each of which has an $M$ antennas BS and a single-antenna UT\footnote{The pilot sequences assigned to different UTs in each cell are assumed to be orthogonal to each other to avoid potential intra-cell interference such that pilot contamination only exists among inter-cell UTs \cite{yinhaifan1}.}. 
	%	in each cell aim to generates symmetric keys, as shown in Fig.~\ref{fig:system_model}. 
	%	 Specifically, , each of which has a  base station (BS)  and multiple single-antenna user terminals (UTs). 
	Under the time-division duplexing (TDD) protocol, the BS and UT in the $k$-th cell, i.e., BS $k$ and UT $k$, $k\in \left\{1,\cdots,K\right\}$, aim to generate a symmetric key ${\kappa}_{k}$ by exploiting the reciprocity of the involved wireless channels. To facilitate the PKG, the system deploys $ L $ RISs with each RIS consisting of $N$ passive reflection elements. 
	Besides, there is a smart controller for coordinating the BSs and adapting the phase shifts of the RISs to enable effective secret key generation \cite{wu2019towards}.
	
	\subsection{RIS-based Channel Model}
	The direct channel between UT $j$ and BS $i$ is denoted as	$\h_{i,j}^d \in \mathbb{C}^{M\times 1}, i,j\in \{1,\cdots,K\}$.
%	The direct channel from the $i$-th BS to the $j$-th UT is denoted as $\h_{i,j}^{d} \in \mathbb{C}^{M \times 1}$, $i,j\in\{1,\cdots,K\}$. 
%	Equipped with a smart controller that communicates with BSs, RISs adapts the phase shift of each reflecting element to enable the secret key generation \cite{wu2019towards}.
	When the RISs are introduced to the PKG system, they establish some indirect additional communication channels. Specifically, the channels from BS $i$ to RIS $l$, $l\in\left\{1,\cdots,L\right\}$, and from RIS $l$ to UT $j$ are denoted as $\G_{i,l}^{r} \in \mathbb{C}^{M \times N}$ and $\h_{l,j}^{r} \in \mathbb{C}^{N \times 1}$, respectively. The diagonal phase-shifting matrix of RIS $l$ is denoted by $\Phim_{l}=\diag\{\v_{l}\}$, where $\v_{l} \in \mathbb{C}^{N\times 1}$ is the phase-shifting vector adopted at RIS $l$.
	Then, the equivalent downlink channel from BS $i$ to UT $j$ is expressed as\footnote{
			Note that different multipath delays incurred by the pilot delays are ignored because the		pilots are assumed to be synchronized to maximize the power of multi-cell interference that serves as a worst case scenario \cite{yinhaifan}. 
%			In addition, the method proposed in this paper is based on the channel covariance and  does not require the accurate pilot synchronization.
		}
	\begin{align}
		\h_{i,j} 
		& = \h_{i,j}^{d} + \sum_{l=1}^{L}\G_{i,l}^{r} \Phim_{l}\h_{l,j}^{r} = \h_{i,j}^{d} + \H_{i,j}^{r}\v, \label{h_ij}
%		\\
%		& = \h_{i,j}^{d} + \G_{i}^{r} \Phim\h_{j}^{r} \\
%		& = \h_{i,j}^{d} + \G_{i}^{r} \diag\{\h_{j}^{r}\} \v \\
%		& = \h_{i,j}^{d} + \H_{i,j}^{r}\v 
%		& = \left[\h_{i,j}^{d}, \H_{i,j}^{r} \right]\left[1, \v^T\right]^T \\
%		& = \H_{i,j}\tilde{\v},
	\end{align}
	where $\H_{i,j}^{r} = \G_{i}^{r} \diag\{\h_{j}^{r}\}$ with $\G_{i,l}^{r} = \left[\G_{i,1}^{r},\cdots, \G_{i,L}^{r}\right]$ and $\h_{j}^{r} = \left[(\h_{l,j}^{r})^H, \cdots, (\h_{l,j}^{r})^H\right]^H $. 
%	$\Phim = \diag\{\v\}$ with 
	The phase-shifting vector at the RISs is $\v=\left[\v_{1}^T, \cdots, \v_{L}^{T}\right]^T$.

	From (\ref{h_ij}), the channel between BS $i$ and UT $j$ depends on the wireless channels and phase-shifting vector $\v$.
	This provides a new DoF for optimizing the key generation performance. Also, the  multi-antenna BSs have the intrinsic capability of signal processing in the spatial domain.	Hence, we  next establish a new PKG framework to exploit the  spatial diversity of the multi-antenna BSs and the environment-controlling characteristic of the RISs to enable the multi-cell secret key generation.
	
%	\blue{We assume that the direct channel $\h_{i,j}^{d}$ follows $\h_{i,j}^{d} \sim \mathcal{CN}(0, \R_{i,j}^d )$. Then, $\h_{i,j}^{d}$ can be reformulated as $\h_{i,j}^{d} = (\R_{i,j}^d)^{\frac{1}{2}} \tilde{\h}_{i,j}^d =
%		 (\R_{i,j}^d)^{\frac{1}{2}} \frac{1}{\sqrt{N}}\tilde{\H}_{i,j}^d \v$, where the entries in  $\tilde{\h}_{i,j}^d$ and $\tilde{\H}_{i,j}^d$  are i.i.d. Gaussian random variables with zero mean and unit variance.
%  }

	\subsection{RIS-aided PKG Framework} 
%		Now, we present a new framework to take full advantage of the RIS-assisted PKG in multicell systems.
		 As shown in Fig. \ref{fig:framework}, the proposed framework consists of three phases. Firstly, during the parameter design phase, the BSs design the precoding matrices and phase-shifting vector. Then, they convey the phase-shifting vector  to the RISs controller to configure the RISs reflection matrix. Secondly, the BSs and UTs acquire the channel estimations by channel probing. Finally, feature-to-key processing is performed to convert the channel estimations into secret keys \cite{2021Sum}, \cite{22Sum-RIS}.
		The specific procedures are shown as follows.
		\begin{figure}
			\centering
			\includegraphics[width=0.53\textwidth]{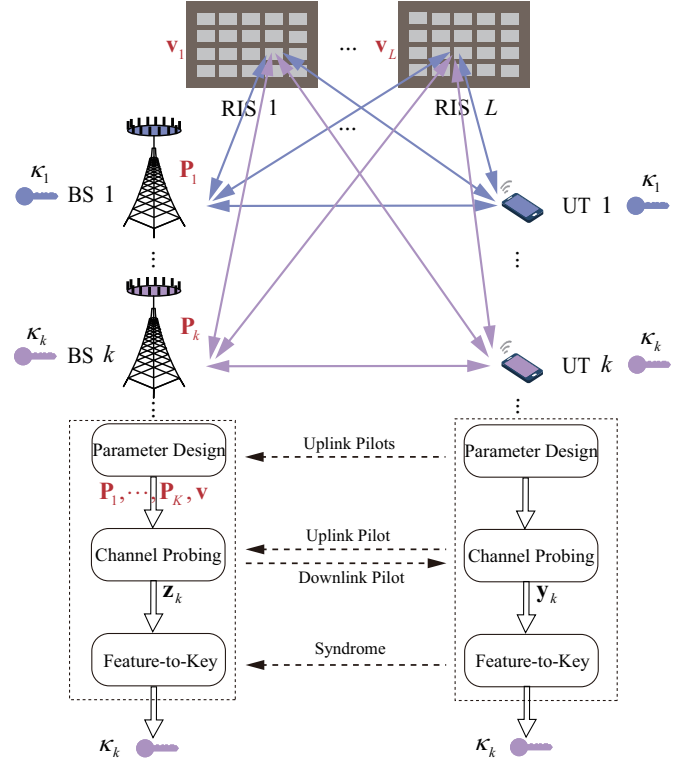}
			\caption{The proposed RIS-aided PKG framework in multi-cell systems.}
			\label{fig:framework}
		\end{figure}
	\begin{enumerate}%[(i)]
		\item \emph{Parameter Design:} In this step, the BSs design the precoding matrices $\P_{k},k\in \left\{1,\cdots,K\right\}$, and phase shifts $\v$ according to some statistical CSI.
		In particular, the UTs in all the cells transmit orthogonal sounding signals such that the BSs can obtain the channel covariance matrices between the BSs and the UTs. Then, the BSs design the parameters $\P_{k}$ and $\v$ via the proposed optimization algorithm, which will be elaborated in Section \ref{sec:algorithm}.
		
		\item \emph{Channel Probing:} In this step, the BSs and UTs probe the channel alternatively and extract the reciprocal channel characteristics with the help of the precoding  matrices and reflection coefficients. 
		In the downlink, each BS transmits the downlink pilot signals processed by the optimized precoding matrix $\P_k$. Then, the UTs estimate the combined channels from the received signals.
		In the uplink, each UT transmits the same pilot signals. The pilot is reflected by the RISs and received by the BSs.
		Then, the BSs estimate the channels from the received signals and utilize the precoding matrices $\P_k$ to form the combined channels.

		\item \emph{Feature-to-Key Processing:} Once the BSs and UTs acquire the channel estimations, they employ some quantization method to generate raw key bits. Finally, information reconciliation and privacy amplification steps are performed to produce the secret keys.
	\end{enumerate}
	\begin{remark}
%		From the KGR (\ref{mutual_information}),
%		Since the KGR is  \cite{}
%		Since the
%		covariance matrices alter slowly in dense scattering environments
		
		In the parameter design step, the prior information to design the precoding matrices and phase shifts are the channel covariance matrices. 
		Due to the fact that the covariance matrices alter slowly across time in dense scattering environments~\cite{yang2020asymptotic}, we can obtain these matrices from the previous several time slots by adopting some existing estimation methods e.g., \cite{yinhaifan}, \cite{6288608}.
%		The precoding matrices and phase shifts are determined by the statistical CSI, which can be obtained by some time and frequency resources. For example, we can use the channel frequency response at several time slots to estimate the channel covariance matrices.
		 Upon the parameter design phase is completed, the BSs and UTs can perform multiple channel probing rounds within some channel coherence times. Each channel probing round is completed within one coherence time to guarantee the similarity of the uplink and downlink channel estimations. Meanwhile, different channel probing rounds should be performed in different coherence times to introduce randomness to the secret keys.
	\end{remark}
%	Since these latter two steps are similar to those used in existing key generation schemes, we do not pay particular attention to them
%	and focus our attention only on optimizing the RIS.
	In this paper, we first focus on the  channel probing step, where the multi-cell pilot contamination exists and then propose an algorithm to design the precoding matrices and phase shifts to improve the PKG performance.
%	In this paper, we focus on the parameter design and the channel probing steps to investigate how to alleviate the multi-cell pilot contamination in PKG systems.

%	\begin{figure}
%		\centering
%		%	\includegraphics[width=3.2in]{fig/Paper-RIS-beam.pdf}
%		\includegraphics[width=0.6\textwidth]{fig/system_model/channel-probing-uplink.eps}
%		\caption{The block diagram of the uplink channel sounding process.}
%		\label{uplink}
%	\end{figure}
%\begin{figure}
%	\centering
%	%	\includegraphics[width=3.2in]{fig/Paper-RIS-beam.pdf}
%	\includegraphics[width=0.6\textwidth]{fig/system_model/channel-probing-downlink.eps}
%	\caption{The block diagram of the downlink channel sounding process.}
%	\label{downlink}
%\end{figure}

	\subsection{Signal Presentation of the Channel Probing} 
	%In PKG system, Alice and Bob treat the reciprocal wireless channel as common randonmess to extract a symmetric key. 
	%The general key generation process consists four steps, i.e. , while in this paper we focus on the first step to maximize the theoretical key rate. 

	%In this paper, 
	%Take full advantage of the RIS-assisted PKG in multi-antenna systems,
	%%Different from the SISO system that only the RIS coefficients are configured to improve the key generation rate, 
	%we propose to joint design the transmit and reflective beamforming at BS and RIS, respectively. 
	%The general PKG contains four steps, namely channel probing, quantization, information reconciliation, and privacy amplification. Since the last three steps are similar to existing works [access], in this paper we focus on the channel probing step where transmit and reflective beamforming are joint designed to facilitate PKG. 
	In the PKG system, the BSs and UTs first perform channel probing to acquire the reciprocal channel estimation. 
	The process of channel probing is described as follows.
	
	The channel sounding step contains two phases, i.e., the downlink phase and the uplink phase, respectively. 
	In the downlink, BS $k$ transmits the downlink public known pilot $\X \in \mathbb{C}^{M_e \times M_e}$ with $\X\X^H=\I_{M_e}$ and  $M_e \leq M $ is the number of radio-frequency (RF) chains.
%	where $M_e $ denotes the number of transmitted pilots. 
	Then, the signals received at UT $k$ is
	\begin{align}
		\left(\r_{k}^d\right)^T = \h_{k,k}^T \P_{k}^T \X
		+ \underbrace{\sum_{i=1,i\neq k}^{K}\h_{i,k}^T \P_{i}^T \X}_{\rm Signals\ from\ BSs\ in\ other\ cells} 
		+ (\n_{k}^d)^T,
	\end{align}
	where $\P_{i} \in \mathbb{C}^{M_e \times M }$ represents the precoding matrix at BS $i$ \cite{2021Sum}, \cite{Sun_access},
	$\n_{k}^d \in \mathbb{C}^{M_e\times 1}$ is the complex Gaussian noise at UT $k$ with zero mean and unit variance. 
	After the standard least-squares (LS) channel estimation~\cite{22Sum-RIS}, \cite{21JiTVT},
	 UT $k$ obtains 
	\begin{align}
		\y_{k} =\X^{*}\r_{k}^d = \P_{k}\h_{k,k}  + \sum_{i=1}^{K}\P_{i}\h_{i,k}  + \X^* {\n}_{k}^{d} .\label{yk}
	\end{align}

	In the uplink phase, the UTs simultaneously transmits the pilot $s $ with $|s|^2 = 1$ and the signal received at BS $k$ is expressed as 
%	in the uplink, UT $k$ transmits the publicly known pilot signal $s \in \mathbb{C}$ with $|s|^2 = 1$, and the received signal at BS $k$ is given by
	\begin{align}
		\r_{k}^{u} = \h_{k,k} s 
		+ \underbrace{\sum_{j=1,j\neq k}^{K}\h_{k,j} s}_{\rm Signals\ from\ UTs\ in\ other\ cells} 
		+ \n^{u}_{k}, 
	\end{align}
	where  $\n^{u}_{k}$ is the complex Gaussian noise zero mean and unit variance.
	Then, BS $k$ performs the LS channel estimation as
	\begin{align}
		\tilde{\z}_{k} 
		= \h_{k,k} +\sum_{j=1,j\neq k}^{K}\h_{k,j} +  s^*{\n}_{k}^{u}. \label{tilde_zk}
	\end{align} 
	Then, BS $k$ multiplies the channel estimation with the precoding matrix $\P_{k}$ that yields
	\begin{align}
		\z_{k} = \P_{k}\tilde{\z}_{k} 
		=  \P_{k}\h_{k,k} +\P_{k}\sum_{j=1,j\neq k}^{K}\h_{k,j} +  \P_{k} s^*{\n}_{k}^{u}. \label{zk}
	\end{align}

%	It can be seen that although the uplink and the downlink channel between BS $k$ and UT $k$ are identical, the interference terms are not perfectly reciprocal. In particular, in the $k$-th cell, the interference in the uplink, i.e., $\sum_{j=1,j\neq k} \h_{k,j}$, is the aggregated channels from the UT $j$ to the BS $k$, while in the downlink, the interference, i.e., $\sum_{i=1,i\neq k} \h_{i,k}$, is the superimposed channels from the BS $i$ to UT $k$.

	It can be observed from (\ref{yk}) and (\ref{zk}) that the extracted channel features at BS $k$ and UT  $k$ include the reciprocal component $\P_{k} \h_{k,k}$, the inter-cell interference, and as well as the noise component. These interference terms are not perfectly reciprocal. In particular, in the $k$-th cell, the interference in the uplink, i.e., $\P_{k}\sum_{j=1,j\neq k} \h_{k,j}$, is the product of $\P_{k}$ and the aggregated channels from UT $j$ to BS $k$, while in the downlink, the interference, i.e., $\sum_{i=1,i\neq k} \P_{i}\h_{i,k}$, is the superimposed channels from BS $i$ to UT $k$. 
%	Therefore, the crux of the multi-cell PKG is to effectively mitigate the pilot contamination in the uplink and downlink. 
%	Next, we propose a RIS-aided PKG scheme to solve this problem.
	Fortunately, these components are influenced by the precoding matrices $\P_{i},i\in\{1,\cdots,K\}$, at the BSs and the RIS phase-shifting vector $\v$. 
	Next, we will analyze the impact of $\P_{i}$ and $\v$ on KGR.
%	focus on the design of
%	 $\P_{i},i\in\{1,\cdots,K\}$ and $\v$.
%	 The feature-to-key process will be performed to convert $\z_{k}$ and $\y_{k}$ into a pair of secret keys, $\mathcal{K}_k$, at the BS $k$ and UT $k$.
	
%	However, the reciprocity of $\y_{k}$ and $\z_{k}$ 

	\section{Secret Key Rate Analysis}
	In this section, based on the channel estimations shown in Section \ref{sec:framework}, 
	we derive the closed-form KGR expression and formulate an optimization problem with respect to (w.r.t.) the variables $\P_{i},i\in\{1,\cdots,K\}$, and $\v$. Then, to facilitate the design, we derive an upper bound of the KGR and formulate an upper bound maximization problem.
	
	\subsection{KGR Analysis}
%	The KGR of the UT $k$ and BS $k$ is defined as the minimum mutual information between their channel estimations, which is expressed as
%	\begin{align}
%		R_{k} = \min_{i\neq k}\mathcal{I}(\y_{k};\z_{k}|\y_{i},\z_{i} ).
%	\end{align}
%	Furthermore, the channels that different UTs and BSs experience are independent since the distance betwen them are assumed to be more than half a wavelength \cite{li2018high}. Hence, the KGR can be simplified as
%	\begin{align}
%		R_{k} = \mathcal{I}(\y_{k};\z_{k} ).
%	\end{align}
	To evaluate the key generation performance, we derive the KGR expression in the following. The KGR is a comprehensive metric to evaluate the consistency, rate, and randomness of the secret key \cite{2009Key}.
	The KGR between UT $k$ and BS $k$ is defined as the mutual information between their channel estimations, which is expressed as
	\begin{align}
		R_{k} = \mathcal{I}(\y_{k};\z_{k} ).\label{mutual_information}
	\end{align}
	Given the combined channel gains $\y_{k}$ and $\z_{k}$ in (\ref{yk}) and (\ref{zk}), the closed-form expression of the KGR can be derived as follows.
	\begin{theorem}
		The KGR of UT $k$ and BS $k$ is given by
		\begin{align}
			R_k & =\ln \det\left( 
			\P_{k}\left[\sum_{j=1}^{K}\R_{k,j}^{d} +
			\left(\bar{\v}^H \otimes \I_{M} \right)\sum_{j=1}^{K}
			\R_{k,j}^{r}\left(\bar{\v} \otimes \I_{M} \right)+\I_{M}\right]\P_{k}^H
			 \right) \notag \\
			& \quad -  \ln \det\left( 
			\P_{k}\left[\sum_{j=1}^{K}\R_{k,j}^{d} +
			\left(\bar{\v}^H \otimes \I_{M} \right)\sum_{j=1}^{K}
			\R_{k,j}^{r}\left(\bar{\v} \otimes \I_{M} \right)+\I_{M}\right]\P_{k}^H \right. \notag \\
			& \quad \left.
			 - \P_{k}\left(\R_{k,k}^{d} + \left(\bar{\v}^H \otimes \I_{M}\right)
			 \R_{k,k}^{r}\left(\bar{\v} \otimes \I_{M}\right)\right)\P_{k}^H 
			 \right. \notag \\
			 & \quad \left. \times 
			 \left(\sum_{i=1}^{K}\P_{i}\left(\R_{i,k}^{d} + \left(\bar{\v}^H \otimes \I_{M}\right)
			 \R_{i,k}^{r}\left(\bar{\v}^H \otimes \I_{M}\right)^H \right) \P_{i}^H +\I_{M_e}\right)^{-1} \right. \notag \\
			 & \quad \left. \times \P_{k}\left(\R_{k,k}^{d} + \left(\bar{\v}^H \otimes \I_{M}\right)
			 \R_{k,k}^{r}\left(\bar{\v} \otimes \I_{M}\right)\right)\P_{k}^H 
			 \right), \label{Th1}
		\end{align}
		where $\R_{i,j}^d \defeq \mathbb{E} \left\{ \h_{i,j}^d \left(\h_{i,j}^d \right)^H \right\}$ and $\R_{i,j}^r \defeq \mathbb{E} \left\{ {\rm vec}\left(\H_{i,j}^r \right) {\rm vec}\left(\H_{i,j}^r \right)^H \right\}$ are the channel covariance of the direct channel and the RISs-involved channel between BS $i$ and UT $j$, respectively. The variable is defined as $\bar{\v} \defeq \v^*$ for the sake of notational simplicity.
%		\begin{align}
%			R_{k} =& - \log_{2} \det\left(\I -  \W_{k}\R_{k,k}\W_{k}^H  \left(\sum_{i=1}^{K}\W_{i}\R_{i,k}\W_{i}^H +  \I\right)^{-1} \W_{k}\R_{k,k}\W_{k}^H\right. \notag \\
%			&\left.  
%			\times  \left(\W_{k} \sum_{j=1}^{K}\R_{k,j}\W_{k}^H + \P_{k}\P_{k}^H\right)^{-1} \right),
%		\end{align}
%	where $\W = \tilde{\v}^T \otimes \P_{k}$.
	\end{theorem}
	\begin{proof}
		Please see Appendix~\ref{sec:proof-KGR-RIS}.
	\end{proof}
	\begin{remark}
		When the precoding matrices at the BSs are not optimized, i.e., $\P_{i} = \sqrt{\frac{P_i}{M}} \I_{M} $, and the RISs are not deployed, the KGR between BS $k$ and UT $k$ is given by
		\begin{align}
			R_{k} &= \ln \det\left( \sum_{j=1}^{K}\R^d_{k,j}+\I_{M}
			\right) - \ln \det\left( \sum_{j=1}^{K}\R^d_{k,j}+\I_{M}
			- {\frac{P_{k}}{M}}\R^d_{k,k}\left(\sum_{i=1}^{K}{\frac{P_{i}}{M}} \R^d_{i,k}+\I_{M}\right)^{-1}\R^d_{k,k}
			\right),\label{KGR-noRIS0}
		\end{align}
	where $P_i$ denotes the transmit power at BS $i$. Furthermore, we present a proposition of the KGR performance as follows.
	
%	\end{remark}
	\begin{proposition}
		$R_k$ decreases monotonically  with increasing $P_i, i\neq k$, and the eigenvalues of $\R^d_{i,k}$ and $ \R^d_{k,j}, j\neq k$. Also, $R_k$ approaches $0$ when these values are sufficiently large.
	\end{proposition}
	\begin{proof}
		Please see Appendix~\ref{KGR-analysis}.
	\end{proof}
%	\begin{remark}
			Proposition 1 proves that the KGR between BS $k$ and UT $k$ decreases with the transmit power of BSs in other cells, i.e., $P_i, i\neq k$, and the channel gain of the interference channels, i.e., $\lambda_{\ell}\left(\R^d_{ i,k }\right), i\neq k$, and $\lambda_{\ell}\left(\R^d_{ k,j }\right), j\neq k$, $\ell \in \left\{1,\cdots,M\right\}$. 
		In addition, $R_k$ tends to $0$ when these interference is sufficiently large. Therefore, the crux of the multi-cell PKG is to effectively mitigate the pilot contamination in the uplink and downlink.
	\end{remark}
	\begin{remark}
		From  (\ref{Th1}), the KGR $R_k$ is decided by the statistical channel information (CSI), i.e., the channel covariance matrices, and the parameters, i.e., the precoding matrices and the phase-shifting vector. 
%		This means the parameter design is based on the statistical CSI.		
		In this paper, we assume that the covariance matrices are perfectly known by the BSs in the parameter design phase\footnote{The impact of imperfect statistical CSI on PKG performance is left to future work.}
%		\footnote{The statistical CSI can be obtained by some existing works, such as []. The impact of the imperfect CSI on PKG performance is postponed to a future work.}
		 and we focus on 
		the design of precoding matrices and phase-shifting vector.
	\end{remark}
	
	Now, we jointly design the precoding matrices and phase shifts to improve the key generation performance. Specifically, we aim for maximizing the WSKR of all the cells by jointly optimizing the $\P_{i}, i\in \{1,\cdots,K\}$, and $\bar{\v}$, while guaranteeing the total power constraint of each BS and the unit modulus of the reflection coefficient of each RIS element. Specifically, the optimization problem can be formulated as
	\begin{align}
		\underset{\P_{i},i\in\{1,\cdots,K\}, {\bar{\v}}}{\text{maximize}}
%		\text{maximize}_{\P_{i},i\in\{1,\cdots,K\}, {\bar{\v}}} 
		\
		&
		%		\ g(\P_{i}, {\bar{\v}}) =
		\sum_{k=1}^{K}w_{k} R_{k} \notag \\
		\ \text { s.t. } 
		&\ \text {C1}{:}\,||\P_{i}||_{F}^2 \leq P_\mathrm{A} , \forall i\in \{1,\cdots,K\}, \notag \\
%		&\ \text {C2}{:}\,\text{rank}\left(\P_{i}\right) =M_e , \forall i\in \{1,\cdots,K\}, \notag \\
		%		&\ \text {C2}{:}\, \text{rank}\left(\P_{i}\right) = M_e , \forall i\in \{1,\cdots, K\}, \notag \\
		&\ \text {C2}{:}\, |v_{n} |\ = 1 , \forall n\in \{1,\cdots, NL\},\label{problem-0}
	\end{align}
	where $w_k \geq 0$ denotes the weight for the $k$-th cell that represents the priority of the corresponding BS and UT. Constraint $\text{C1}$ represents the power of the precoding matrices at each BS should be less than the maximum transmit power $P_{\rm A}$ \cite{Sun_access}. 
	%	Constraint $\text{C2}$ guarantees that $\z_{k}$ is multivariate Gaussian vector and its probability density function (PDF) exists. Finally, 
%	Constraint $\text{C2}$ represents the precoding matrices should be row full rank to guarantee the pdf of $\z_{k}$ and $\y_{k}$ exists and the $\ln \det (\cdot)$ in the first term of $R_k$ will not tend to minus infinity.
	Constraint $\text{C2}$ represents the  modulus constraint of each
	reflection coefficient at the RISs \cite{Yu_JSAC}.
	
	Note that the objective function WSKR is non-convex w.r.t. $\P_i$ and $\bar{\v}$ since these variables are coupled in the
	matrix inversion operation. Additionally, the unit-modulus constraint $\text{C2}$ is also non-convex. Hence, Problem (\ref{problem-0}) is generally NP-hard and 
	difficult to find the globally optimal solution \cite{bubeck2015convex}. Next, as a compromise, we derive an upper bound of the KGR and formulate an upper bound maximization problem.

	\subsection{Upper Bound of KGR}
	Recalling (\ref{tilde_zk}) and (\ref{zk}), the channel estimation and the combined channel at BS $k$ are $\tilde{\z}_{k} 
	= \sum_{j=1}^{K}\h_{k,j} + s^*{\n}_{k}^{u}$ and $\z_{k} = \P_{k} \tilde{\z}_{k}$, respectively. 
	Therefore,  $\y_{k}$ and ${\z}_{k}$ are conditionally independent given $\tilde{\z}_{k}$, which means $\y_{k}$, $\tilde{\z}_{k}$, and $\z_{k}$ form a Markov chain $\y_{k} \rightarrow \tilde{\z}_{k} \rightarrow \z_{k}$. According to the data-procesing inequality~\cite{cover1999elements}, an upper bound of the KGR can be expressed as $R_k = \mathcal{I}(\y_{k};\z_{k}) \leq \mathcal{I}(\y_{k};\tilde{\z}_{k}) \defeq R_k^{\rm ub}$. Then, the closed-form expression of $R_k^{\rm ub}$ can be derived as follows.

	%Therefore, we can obtain the closed expression for the upper bound as follows.
	\begin{theorem}\label{Th2}
		The KGR of $k$-th cell, $R_{k}$, is upper bounded by
		\begin{align}
			R_k^{\rm  ub} = \ln \det \left(\sum_{i=1}^{K}\P_{i}\M_{i,k} \P_{i}^H +\I_{M_e}\right) - 
			\ln \det\left(\sum_{i=1}^{K}\P_{i}\M_{i,k} \P_{i}^H +\I_{M_e} -  
			\P_{k} \N_{k,k} \P_{k}^H  \right), \label{KGR:ub}
		\end{align}
		%where  $\M_{i,k} = \R_{i,k}^{d} + \left(\bar{\v}^H \otimes \I_{M}\right)
		%\R_{i,k}^{r}\left(\bar{\v}^H \otimes \I_{M}\right)^H$ and $\N_{k,k}  = \left(\R_{k,k}^{d} + \left(\bar{\v}^H \otimes \I_{M}\right)
		%\R_{k,k}^{r}\left(\bar{\v} \otimes \I_{M}\right)\right)
		%\R_{\tilde{\z}_{k}}^{-1} \left(\R_{k,k}^{d} + \left(\bar{\v}^H \otimes \I_{M}\right)
		%\R_{k,k}^{r}\left(\bar{\v} \otimes \I_{M}\right)\right)$.
		where 
		\begin{align}
			\M_{i,k} &= \R_{i,k}^{d} + \left(\bar{\v}^H \otimes \I_{M}\right)
			\R_{i,k}^{r}\left(\bar{\v}^H \otimes \I_{M}\right)^H, \\
			\N_{k,k} & = \left(\R_{k,k}^{d} + \left(\bar{\v}^H \otimes \I_{M}\right)
			\R_{k,k}^{r}\left(\bar{\v} \otimes \I_{M}\right)\right)
			\R_{\tilde{\z}_{k}}^{-1} \left(\R_{k,k}^{d} + \left(\bar{\v}^H \otimes \I_{M}\right)
			\R_{k,k}^{r}\left(\bar{\v} \otimes \I_{M}\right)\right).
		\end{align}
	\end{theorem}
	\begin{proof}
		Please see Appendix~\ref{sec:proof-KGR-upperbound}.
	\end{proof}
	\begin{remark}
		We next prove  the tightness of the derived upper bound. Specifically, the sufficient condition for achieving the upper bound, i.e., $R_k^{\rm ub} = R_k$, is $\mathcal{I}(\y_{k};\tilde{\z}_k | \z_{k}) = 0$. In other words, $\y_{k}$, $\z_{k}$, and $\tilde{\z}_{k}$ also form a Markov chain $\y_{k}  \rightarrow \z_{k}\rightarrow \tilde{\z}_{k} $ \cite{cover1999elements}. Therefore, one typical case for the equality hold is $M_e = M$ such that $\mathcal{I}(\y_{k};\tilde{\z}_k | \P_{k}\tilde{\z}_k) = \mathcal{H}(\tilde{\z}_k|\P_{k}\tilde{\z}_k) -\mathcal{H}(\tilde{\z}_k|\y_{k},\P_{k}\tilde{\z}_k)= 0$. 
		This means when the number of RF  chains at the BS is the same as the number of antennas, the derived upper bound is tight and equal to the actual KGR. 
%		In this case, the $\z_{k} = \P_{k}\tilde{\z}_{k}$ contains the full information of $\tilde{\z}_{k}$ and thus, $\mathcal{I}(\y_{k};\tilde{\z}_k ) = \mathcal{I}(\y_{k}; \z_{k})$.
		
	\end{remark}
	
	Thanks to Theorem \ref{Th2}, the objective function is significantly simplified. Then, the upper bound optimization problem can be expressed as
	\begin{align}
		\underset{\P_{i},i\in\{1,\cdots,K\}, {\bar{\v}}}{\text{maximize}} \
		&\ \sum_{k=1}^{K}w_{k} R_k^{\rm ub}
		%	\left(\log \det \left(\sum_{i=1}^{K}\P_{i}\M_{i,k} \P_{i}^H +\I_{M}\right) - 
		%	\log \det\left(\sum_{i=1}^{K}\P_{i}\M_{i,k} \P_{i}^H +\I_{M} -  
		%	\P_{k} \N_{k,k} \P_{k}^H  \right)\right)
		\notag \\
		\ \text { s.t. } 
		&\ \text {C1},\text {C2}. \label{problem-01}
	\end{align}
	It can be seen that the  variables $\P_{i},i=1,\cdots,K$, and $\bar{\v}$ are still coupled in the objective function of Problem (\ref{problem-01}). 
	In the next section, we apply an AO-based algorithm to alternately solve for $\left\{\P_{1},\cdots, \P_{K}\right\}$ and $\bar{\v}$ while fixing one of the variable sets.

\section{AO-based Algorithm for Maximum Upper Bound} \label{sec:algorithm}
In this section, we provide an AO-based algorithm to tackle Problem (\ref{problem-01}).
%	optimize the precoding matrices at BSs and the phase shifts at RISs.
%	 to maximize the upper bound of WSKR.
%focus on solving the formulated optimization
%problem. 
%Firstly, we  derive an upper bound of the WSKR and formulate an upper bound maximization problem. 
Specifically, to tackle the coupling of the optimization  variables, we divide the joint problem into two sub-problems, each of which optimizes $\left\{\P_{1},\cdots,\P_{K} \right\}$ or $\bar{\v}$. 
For the sub-problem for $\left\{\P_{1},\cdots,\P_{K} \right\}$, we apply a Lagrangian multiplier method based on KKT conditions, while the sub-problem for $\bar{\v}$ is solved by the PGA algorithm.

\subsection{Optimization of the Precoding Matrices at the BSs}
We first present the optimization of precoding matrices $\left\{\P_{1},\cdots,\P_{K}\right\}$ for given $\bar{\v}$.
%${\P_{1},\cdots, \P_{i-1},\P_{i+1},\cdots, \P_{K}}$
By denoting $\P = \left[\P_{1},\cdots,\P_{K}\right]$, the sub-problem of $\P$ optimization is given by
\begin{align}
\underset{\P }{\text{maximize}} \
	&\ \bar{g}\left(\P\right) =
	\sum_{k=1}^{K}w_{k} \left(\ln \det \left(\P \M_{k} \P^H +\I_{M_e}\right) - 
	\ln \det\left(\P\N_{k} \P^H +\I_{M_e}  \right)\right)
%	\sum_{k=1}^{K}w_{k} \left(\ln \det \left(\sum_{i=1}^{K}\P_{i}\M_{i,k} \P_{i}^H +\I_{M}\right) - 
%	\ln \det\left(\sum_{i=1}^{K}\P_{i}\M_{i,k} \P_{i}^H +\I_{M} -  
%	\P_{k} \N_{k,k} \P_{k}^H  \right)\right)
	 \notag \\
	\ \text { s.t. } 
	&\ \text {C1},   \label{subproblem-p}
\end{align}
where $\M_k = \diag \left(\M_{1,k},\cdots,\M_{K,k}\right)$ and 
$\N_k = \diag \left(\M_{1,k},\cdots,\M_{k,k} - \N_{k,k},\cdots,\M_{K,k}\right) $.
Problem (\ref{subproblem-p}) is still non-convex since the objective function is in high-order w.r.t. $\P$. To derive an efficient suboptimal solution, we employ the Lagrangian multiplier method to address this problem.
Specifically, by introducing the Lagrange multipliers $\lambda_{i}, i=1,\cdots,K$, the Lagrangian function of Problem (\ref{subproblem-p}) is 
\begin{align}
	\mathcal{L}\left(\P_{i},\lambda_{i}\right)= \sum_{k=1}^{K}w_{k} \left(\ln \det \left(\P \M_{k} \P^H +\I_{M_e}\right) - 
	\ln \det\left(\P\N_{k} \P^H +\I_{M_e}  \right)\right) -     \tr \left[\P {\mat \Lambda}  \P^H\right]  + P_{\rm A} \sum_{i=1}^{K} \lambda_{i},
\end{align}
where ${\mat \Lambda} = \diag \left(\lambda_{1}\I_{M},\cdots, \lambda_{K}\I_{M}\right)$
%\begin{align}
%	\mathcal{L}\left(\P_{i},\lambda_{i}\right) &= \sum_{k=1}^{K}w_{k} \left(\ln\det \left(\sum_{i=1}^{K}\P_{i}\M_{i,k}\P_{i}^H + \I_{M} \right) -  \ln\det \left(\I_{M} + 
%	\sum_{i=1}^{K}\P_{i}\M_{i,k}\P_{i}^H -\P_{k}\N_{k,k}\P_{k} ^H \right)\right)
%	\notag \\
%	&\quad - \lambda_{i}
%	\left(\tr(\P_{i}\P_{i}^H )- P_{\rm A}
%	\right),
%\end{align}
and $\lambda_{i} \geq 0, i\in \left\{1, \cdots, K\right\}$ is associated with the power constraint in $\text{C1}$ at BS $i$. 
%Then, we can derive the KKT condition of Problem (\ref{subproblem-p}) as follows.
By analyzing the KKT condition, we can characterize the structure of the suboptimal solution to $\P$ in the following theorem.
%can obtain a necessary condition that a candidate solution of $\P$ need to fulfill.
%a necessary condition of a candidate solution
\begin{theorem}\label{th3}
	The KKT condition for Problem (\ref{subproblem-p}) is given by
	\begin{align}
		\rm{vec}\left(\P\right) 
		&= \left({\mat \Lambda}^T \otimes\I_{M_e}
		+\sum_{k=1}^{K}w_{k}\left(\N_{k}^T 
		\otimes \left(\I_{M_e} + 
		\P\N_{k}\P^H \right)^{-1}\right) 
		\right)^{-1} \notag \\
		&\quad \times \sum_{k=1}^{K}w_{k}
		\left(
		\left(\M_{k}^T \otimes
		\left(\P\M_{k}\P^H + \I_{M_e} \right)^{-1}\right) 
		\rm{vec}\left(\P\right)
		\right). \label{KKT}
	\end{align}
%	\begin{align}
%		\text{vec}\left(\P_{i}\right) &= \left(\lambda_{i}\I_{M^2}
%		+\sum_{k=1}^{K}w_{k}\left(\L_{i,k}^T 
%		\otimes \left(\I_{M} + 
%		\sum_{i=1}^{K}\P_{i}\M_{i,k}\P_{i}^H -\P_{k}\N_{k,k}\P_{k} ^H \right)^{-1}\right) 
%		\right)^{-1} \notag \\
%		&\quad \times \sum_{k=1}^{K}w_{k}
%		\left(
%		\left(\M_{i,k}^T \otimes
%		\left(\sum_{i=1}^{K}\P_{i}\M_{i,k}\P_{i}^H + \I_{M} \right)^{-1}\right) 
%		\text{vec}\left(\P_{i}\right)
%		\right), \label{KKT}
%	\end{align}
%where the matrix $\L_{i,k}$ is defined as
%\begin{align}
%	\L_{i,k}	\defeq \left\{\begin{array}{cl}
%		\M_{i,k} , & i \neq k, \\
%		\M_{i,k}-\N_{k,k}, & i=k.
%	\end{array}   \right. 
%\end{align}
\end{theorem}
\begin{proof}
	Please see Appendix~\ref{sec:proof-KKT}.
\end{proof}

Next, the value of $\lambda_{i},i\in \left\{1,\cdots,K\right\}$, should  be chosen for ensuring the power constraint at each BS in (\ref{problem-01}) is satisfied, which means
\begin{align}
	\lambda_{i}
	\left(\tr(\P_{i}\P_{i}^H )- P_{\rm A}
	\right) = 0.
\end{align}
To this end, recalling the KKT condition in (\ref{KKT}), we can rewrite the $\text{vec}\left(\P\right)$ as
	\begin{align}
		\text{vec}\left(\P\right) 
		&= \left(\diag \left(\lambda_{1} \I_{MM_e},\cdots,\lambda_{K} \I_{MM_e}\right) 
		+ \F \right)^{-1} \a \\
%		& = \left(\diag \left(\lambda_{1} \I_{M^2} + \F_{1},\cdots,\lambda_{K} \I_{M^2} + \F_{K}\right) 
%		\right)^{-1} \a \\
		& = \diag \left(\left(\lambda_{1} \I_{MM_e} + \F_{1}\right)^{-1},\cdots,\left(\lambda_{K} \I_{MM_e} + \F_{K}\right)^{-1}
		\right) \a,
	\end{align}
	where 
	\begin{align}
		\F = \sum_{k=1}^{K}w_{k}\left(\N_{k}^T 
		\otimes \left(\I_{M_e} + 
		\P\N_{k}\P^H \right)^{-1}\right) = \diag \left(\F_{1}, \cdots, \F_{K} \right) ,
	\end{align}
	and $\a = \sum_{k=1}^{K}w_{k}
	\left(
	\left(\M_{k}^T \otimes
	\left(\P\M_{k}\P^H + \I_{M_e} \right)^{-1}\right) 
	\rm{vec}\left(\P\right)
	\right) = \left[\a_{1}^T, \cdots, \a_{K}^T\right]^T$.	Thus, we can express the $\P_{i}$ at BS $i$ as 
	\begin{align}
		\text{vec}\left(\P_{i}\right) = \left(\lambda_{1} \I_{MM_e} + \F_{i}\right)^{-1} \a_{i}.
	\end{align}
%	and the value of $\left(\text{vec}\left(\P_{i}\right)\right)^H \text{vec}\left(\P_{i}\right)$ decrease monotonically with respect to $\lambda_{i}$

%	Next, we elaborate on how to obtain the optimal $\lambda_{i}$.
%	Recalling KKT condition (\ref{KKT}) for $\P_{i}$, 	we can rewritten the $\text{vec}\left(\P_{i}\right)$ as
%	\begin{align}
%		\text{vec}\left(\P_{i}\right) 
%		&= \left(\lambda_{i}\I_{M} + \F_{i}\right)^{-1} \a_{i},
%		& = \left(\lambda_{i}\I_{M} + \D_{i} \diag \left(\f_{i}\right)
%		\D_{i}^H
%		\right)^{-1} \a_{i} \\
%		& = \D_{i}\left(\lambda_{i}\I_{M} +  \diag \left(\f_{i}\right)
%		\right)^{-1}\D_{i}^H \a_{i}.
%	\end{align}
%where 
%\begin{align}
%	\F_{i} = \sum_{k=1}^{K}w_{k}\left(\L_{i,k}^T 
%	\otimes \left(\I_{M} + 
%	\sum_{i=1}^{K}\P_{i}\M_{i,k}\P_{i}^H -\P_{k}\N_{k,k}\P_{k} ^H \right)^{-1}\right) ,
%\end{align}
%and $\a_{i} = \sum_{k=1}^{K}w_{k}
%\left(
%\left(\M_{i,k}^T \otimes
%\left(\sum_{i=1}^{K}\P_{i}\M_{i,k}\P_{i}^H + \I_{M} \right)^{-1}\right) 
%\text{vec}\left(\P_{i}\right)
%\right)$. 
By performing eigenvalue decomposition on matrix $\F_{i}$ as $\F_{i} = \D_{i} \diag \left(\f_{i}\right)
\D_{i}^H$, we have 
	\begin{align}
		\phi_{i}(\lambda_{i}) &= \text{vec}\left(\P_{i}\right)^H \text{vec}\left(\P_{i}\right) \\
		&= 
		\a_{i}^H \D_{i}\left(\lambda_{i}\I_{MM_e} +  \diag \left(\f_{i}\right)
		\right)^{-1}\left(\lambda_{i}\I_{MM_e} +  \diag \left(\f_{i}\right)
		\right)^{-1}\D_{i}^H \a_{i}  = \sum_{m=1}^{MM_e}\frac{b_{i,m}^2}{\left(\lambda_{i} + f_{i,m}\right)^2},
	\end{align}
where $\f_{i} = [f_{i,1},\cdots, f_{i,MM_e}]^T$ and $\D_{i}^H \a_{i} = [b_{i,1},\cdots, b_{i,MM_e}]^T$. Since $f_{i,m} \geq 0$, it can be verified that $f_{i}(\lambda_{i})$ is a monotonically decreasing function w.r.t. $\lambda_{i}$. 
%Thus, if $f_{i}(0) \leq P_{\rm A}$, the optimal 

Therefore, if $\phi_{i}(0) \leq P_{\rm A}$, the optimized $\P_{i}$ is $\text{vec}\left(\P_{i}^{\rm KKT}\right) = \F_{i}^{-1}\a_{i}$ when $\lambda_{i} $ is set to 0. On the other hand, if $\phi_{i}\left(0\right) > P_{\rm A}$,  there exists a $\bar{\lambda}_{i}$  that satifies $\phi_{i}\left(\bar{\lambda}_{i} \right) = P_{\rm A}$. To find  $\bar{\lambda}_{i}$, we can apply the bisection based search method and the upper bound of $\lambda_i$ is set as 
\begin{align}
	\lambda_{i}^{\rm ub} = \sqrt{\frac{\sum_{m=1}^{M} b_{i,m}^2 }{P_{\rm A}}},
\end{align}
since 
\begin{align}
	\phi_{i}(\lambda_{i}) < \sum_{m=1}^{M}\frac{b_{i,m}^2}{\left(\lambda_{i}^{\rm ub} \right)^2} = P_{\rm A}.
\end{align}

	The overall algorithm to optimize $\P$ and $\lambda_{i},i\in\left\{1,\cdots,K\right\}$, is summarized as Algorithm 1.
	\begin{algorithm}[h]
		\caption{The Lagrangian Dual Algorithm for Problem (\ref{subproblem-p}).}
		\label{alg:1}
		\begin{algorithmic}[1]
			\Require $w_k, \M_{k}$, $\N_{k}, k=1,\cdots,K$, $P_{\rm A}, \epsilon$.
			\State 
			Set: $t = 0$ (iteration index).
			\State Initial:  $\P^{(0)}$.	
			\Repeat  
			\State Calculate ${\P}^{(t+1)}$ as
%			\begin{align}
%				\text{vec}\left(\P_{i}^{(t+1)}\right) &= \left(\lambda_{i}\I_{M^2} 
%				+\sum_{k=1}^{K}w_{k}\left(\L_{i,k}^T 
%				\otimes \left(\I_{M} + 
%				\sum_{i=1}^{K}\P_{i}^{(t)}\M_{i,k}\left(\P_{i}^{(t)}\right)^{H} -\P_{k}^{(t)}\N_{k,k}\left(\P_{k} ^{(t)}\right)^{H} \right)^{-1}\right) 
%				\right)^{-1} \notag \\
%				&\quad \times \sum_{k=1}^{K}w_{k}
%				\left(
%				\left(\M_{i,k}^T \otimes
%				\left(\sum_{i=1}^{K}\P_{i}^{(t)}\M_{i,k}\left(\P_{i}^{(t)}\right)^H + \I_{M} \right)^{-1}\right) 
%				\text{vec}\left(\P_{i}^{(t)}\right)
%				\right). \label{eq:Pi}
%			\end{align}
%				\begin{align}
%				\text{vec}\left(\P_{i}^{(t+1)}\right) &= \left(\lambda_{i}\I_{M^2} 
%				+\F_{i}\left(\P_{i}^{(t)}\right) 
%				\right)^{-1} \a_{i}\left(\P_{i}^{(t)}\right) . \label{eq:Pi}
%			\end{align}
			\begin{align}
				{\rm vec}\left(\P^{(t+1)}\right) 
				&= \left({\mat \Lambda}^T \otimes\I_{M_e}
				+\sum_{k=1}^{K}w_{k}\left(\N_{k}^T 
				\otimes \left(\I_{M_e} + 
				\P^{(t)}\N_{k}\left(\P^{(t)}\right)^H \right)^{-1}\right) 
				\right)^{-1} \notag \\
				&\quad \times \sum_{k=1}^{K}w_{k}
				\left(
				\left(\M_{k}^T \otimes
				\left(\P^{(t)}\M_{k}\left(\P^{(t)}\right)^H + \I_{M_e} \right)^{-1}\right) 
			{	\rm vec}\left(\P^{(t)}\right)
				\right). \label{eq:Pi}
			\end{align}
%			\State Use bisection search method to update $\lambda_{i}$;
			\State Set: $i=1$ (BS index).
			\Repeat
			\State Initialize the bounds $\lambda_{i}^{\rm ub}$ and $\lambda_{i}^{\rm lb}$.
			\Repeat
			\State If $\phi_{i}(0) \leq P_A$ holds, $\rm{vec}\left(\P_{i}^{\text{KKT}}\right) = \F_{i}^{-1}\a_{i}$; otherwise, go to the next step.
			\State Calculate $\lambda_{i} = \left(\lambda_{i}^{\rm ub} + \lambda_{i}^{\rm lb}\right)/2$.
			\State If $\phi_{i}(\lambda_{i}) \leq P_A$, set $\lambda_{i}^{\rm ub} = \lambda_{i}$; Otherwise, set $\lambda_{i}^{\rm lb} = \lambda_{i}$.
			\Until{$|\lambda_{i}^{\rm ub} - \lambda_{i}^{\rm lb}| \leq \epsilon$.}
			\Until{$i=K$.}
			\State Set $t = t+1$ and calculate the secret key rate $\bar{g}(\P^{(t+1)})$.
			\Until{$\frac{|\bar{g}(\P^{(t+1)}) -  \bar{g}(\P^{(t)})|}{\bar{g}(\P^{(t)})} \leq \epsilon .$}
		\end{algorithmic}
	\end{algorithm}

\subsection{Optimization of Phase Shifts at the RISs}
Next, we present the optimization of phase shift vector $\bar{\v}$ when $\left\{\P_{1},\cdots,\P_{K}\right\}$ are fixed. To facilitate the algorithm design, we can first formulate the sub-problem for the RISs as 
%reformulate the expression of $R_k^{ub}$ 
%for RISs as 
\begin{align}
	\underset{ {\bar{\v}}}{\text{maximize}} \
	&\ \tilde{g}(\bar{\v}) = \sum_{k=1}^{K}w_{k}R_{k}^{\rm ub}(\bar{\v}) \notag \\
	\ \text { s.t. } 
	&\ \text {C2} ,  \label{subproblem-v}
\end{align}
where 
\begin{align}
	R_{k}^{\rm ub}(\bar{\v}) &= \ln
	\det \left(\L_k\right)  - \ln
	\det \left(\L_k - 
	\left(  \P_{k}\R_{k,k}^d  +
	{\left(\bar{\v}^H \otimes \I_{M_e} \right)}
	\left(\I_{NL} \otimes \P_{k} \right)
	\R_{k,k}^{r}  {\left(\bar{\v}^H \otimes \I_M \right)^H}\right) \right.\notag \\
	& \left.\quad \times
	\left(\sum_{j=1}^{K}\R_{k,j}^{d} +
	{\left(\bar{\v}^H \otimes \I_M \right)}\sum_{j=1}^{K}
	\R_{k,j}^{r}{\left(\bar{\v}^H \otimes \I_M \right)^H}+\I_M\right)^{-1}
	\right.\notag \\
	& \left.\quad \times \left(\R_{k,k}^d \P_{k}^H +
	{\left(\bar{\v}^H \otimes \I_{M} \right)}
	\R_{k,k}^{r} \left(\I_{NL} \otimes \P_{k} \right)^H {\left(\bar{\v}^H \otimes \I_{M_e} \right)^H}\right)
	\right), \label{Rub_v}
\end{align}
%\begin{proof}
%	Please see Appendix~\ref{sec:proof-KGR-RIS-v}.
%\end{proof}
where $\L_k = \sum_{i=1}^{K}\P_{i}\R_{i,k}^{d}\P_{i}^H +
{\left(\bar{\v}^H \otimes \I_{M_e} \right)}
\sum_{i=1}^{K}
\left(\I_{NL} \otimes \P_{i} \right)
\R_{i,k}^{r}\left(\I_{NL} \otimes \P_{i} \right)^H
{\left(\bar{\v}^H \otimes \I_{M_e} \right)^H} + \I_{M_e}$. 
From (\ref{Rub_v}), the objective function is intractable w.r.t. $\bar{\v}$.
The unit modulus constraint $\text{C2}$ is also non-convex and there is no general approach to solve unit modulus constrained 
problems optimally.
Therefore, we adopt the PGA algorithm to find a  stationary solution of Problem (\ref{subproblem-v}) \cite{Papazafeiropoulos}. At each iteration, PGA projects the solution onto the closest feasible point satisfying the unit-modulus constraint.

Specifically, at iteration $t$, we first calculate the conjugate gradient $\nabla_{\bar{\v}^*} g(\bar{\v})$ as the ascent direction to guarantee the increase of the objective function.
The closed-form expression of the gradient is derived as follows.
\begin{theorem}\label{lemma-v}
	The Euclidean gradient of function $g(\bar{\v})$ w.r.t. $\bar{\v}^{*}$ is given by
	\begin{align}
		\nabla_{\bar{\v}^*} g(\bar{\v}) = \sum_{k=1}^{K} w_k \left(\sum_{m=1}^{M_e} \bar{\q}_{k,m}+ \sum_{m=1}^{M}\bar{\g}_{k,m}\right),
	\end{align}
where 
$\bar{\q}_{k,m} = \left[ \Q_{k}\right]_{(m-1)NL+1:mNL,m}$, 
$\bar{\g}_{k,m} = \left[ \G_{k}\right]_{(m-1)NL+1:mNL,m}$ and
%\begin{align}
%	\Q_{k} = K_{2}
%	\sum_{i=1}^{K}
%	\left(\I \otimes \P_{i} \right)
%	\R_{i,k}^{r}\left(\I \otimes \P_{i} \right)^H
%	\left(\bar{\v} \otimes \I \right)
%	\R_{\y_{k}}^{-1} 
%	K_{1},
%\end{align}
%and
\begin{align}
	\Q_{k} 
	& = \mat{K}_{M_eNL}
	\sum_{i=1}^{K}
	\left(\I_{NL} \otimes \P_{i} \right)
	\R_{i,k}^{r}\left(\I_{NL} \otimes \P_{i} \right)^H
	\left(\bar{\v} \otimes \I_{M_e} \right)
	\R_{\y_{k}}^{-1} 
	\mat{K}_{M_e} \notag \\
	&\quad - \mat{K}_{M_eNL}
	\sum_{i=1}^{K}
	\left(\I_{NL} \otimes \P_{i} \right)
	\R_{i,k}^{r}\left(\I_{NL} \otimes \P_{i} \right)^H
	\left(\bar{\v} \otimes \I_{M_e} \right)
	\left(\R_{\y_{k}} -  \R_{\y_{k}\tilde{\z}_{k}}\R_{\tilde{\z}_{k}}^{-1} \R_{\tilde{\z}_{k}\y_{k}} \right)^{-1} \mat{K}_{M_e} \notag \\
	&\quad +\mat{K}_{M_eNL}
	\left(\I_{NL} \otimes \P_{k} \right)
	\R_{k,k}^{r}  \left(\bar{\v} \otimes \I_{M} \right)
	\R_{\tilde{\z}_{k}}^{-1} \R_{\tilde{\z}_{k}\y_{k}} 
	\left(\R_{\y_{k}} -  \R_{\y_{k}\tilde{\z}_{k}}\R_{\tilde{\z}_{k}}^{-1} \R_{\tilde{\z}_{k}\y_{k}} \right)^{-1} 
	\mat{K}_{M_e}, \notag \\
	\G_{k} &= 
%	& = \mat{K}_{MN}
%	\sum_{i=1}^{K}
%	\left(\I \otimes \P_{i} \right)
%	\R_{i,k}^{r}\left(\I \otimes \P_{i} \right)^H
%	\left(\bar{\v} \otimes \I \right)
%	\R_{\y_{k}}^{-1} 
%	\mat{K}_{M} \notag \\
%	&\quad - \mat{K}_{MN}
%	\sum_{i=1}^{K}
%	\left(\I \otimes \P_{i} \right)
%	\R_{i,k}^{r}\left(\I \otimes \P_{i} \right)^H
%	\left(\bar{\v} \otimes \I \right)
%	\left(\R_{\y_{k}} -  \R_{\y_{k}\tilde{\z}_{k}}\R_{\tilde{\z}_{k}}^{-1} \R_{\tilde{\z}_{k}\y_{k}} \right)^{-1} \mat{K}_{M} \notag \\
%	&\quad +\mat{K}_{MN}
%	\left(\I \otimes \P_{k} \right)
%	\R_{k,k}^{r}  \left(\bar{\v} \otimes \I \right)
%	\R_{\tilde{\z}_{k}}^{-1} \R_{\tilde{\z}_{k}\y_{k}} 
%	\left(\R_{\y_{k}} -  \R_{\y_{k}\tilde{\z}_{k}}\R_{\tilde{\z}_{k}}^{-1} \R_{\tilde{\z}_{k}\y_{k}} \right)^{-1} 
%	\mat{K}_{M} \notag \\
	- \mat{K}_{MNL}\sum_{j=1}^{K}
	\R_{k,j}^{r}\left(\bar{\v} \otimes \I_M \right)
	\R_{\tilde{\z}_{k}}^{-1}\R_{\tilde{\z}_{k}\y_{k}}  
	\left(\R_{\y_{k}} -  \R_{\y_{k}\tilde{\z}_{k}}\R_{\tilde{\z}_{k}}^{-1} \R_{\tilde{\z}_{k}\y_{k}} \right)^{-1} 
	\R_{\y_{k}\tilde{\z}_{k}}
	\R_{\tilde{\z}_{k}}^{-1}
	\mat{K}_{M} \notag \\
	&\quad + \mat{K}_{MNL}
	\R_{k,k}^{r} \left(\I_{NL} \otimes \P_{k} \right)^H \left(\bar{\v} \otimes \I_{M_e} \right)
	\left(\R_{\y_{k}} -  \R_{\y_{k}\tilde{\z}_{k}}\R_{\tilde{\z}_{k}}^{-1} \R_{\tilde{\z}_{k}\y_{k}} \right)^{-1} 
	\R_{\y_{k}\tilde{\z}_{k}}\R_{\tilde{\z}_{k}}^{-1} 
	\mat{K}_{M}, \label{grad}
\end{align}
where $\mat{K}_{I} \in \mathbb{C}^{I \times I}, I\in \left\{M_e,M_eNL, MNL, M\right\}$, is the commutation matrix \cite{zhang2017matrix}. 
%
%\begin{align}
%	\G_{k} &= K_{2}
%	\sum_{i=1}^{K}
%	\left(\I \otimes \P_{i} \right)
%	\R_{i,k}^{r}\left(\I \otimes \P_{i} \right)^H
%	\left(\bar{\v} \otimes \I \right)
%	\left(\R_{\y_{k}} -  \R_{\y_{k}\tilde{\z}_{k}}\R_{\tilde{\z}_{k}}^{-1} \R_{\tilde{\z}_{k}\y_{k}} \right)^{-1} K_{1} \notag \\
%	&\quad -K_{2}
%	\left(\I \otimes \P_{k} \right)
%	\R_{k,k}^{r}  \left(\bar{\v} \otimes \I \right)
%	\R_{\tilde{\z}_{k}}^{-1} \R_{\tilde{\z}_{k}\y_{k}} 
%	\left(\R_{\y_{k}} -  \R_{\y_{k}\tilde{\z}_{k}}\R_{\tilde{\z}_{k}}^{-1} \R_{\tilde{\z}_{k}\y_{k}} \right)^{-1} 
%	K_{1} \notag \\
%	&\quad  - K_{2}\sum_{j=1}^{K}
%	\R_{k,j}^{r}\left(\bar{\v} \otimes \I \right)
%	\R_{\tilde{\z}_{k}}^{-1}\R_{\tilde{\z}_{k}\y_{k}}  
%	\left(\R_{\y_{k}} -  \R_{\y_{k}\tilde{\z}_{k}}\R_{\tilde{\z}_{k}}^{-1} \R_{\tilde{\z}_{k}\y_{k}} \right)^{-1} 
%	\R_{\y_{k}\tilde{\z}_{k}}
%	-\R_{\tilde{\z}_{k}}^{-1}
%	K_{1} \notag \\
%	&\quad - K_{2}\sum_{j=1}^{K}
%	\R_{k,j}^{r}\left(\bar{\v} \otimes \I \right)
%	\R_{\tilde{\z}_{k}}^{-1}
%	\left(\R_{\y_{k}} -  \R_{\y_{k}\tilde{\z}_{k}}\R_{\tilde{\z}_{k}}^{-1} \R_{\tilde{\z}_{k}\y_{k}} \right)^{-1} 
%	\R_{\y_{k}\tilde{\z}_{k}}\R_{\tilde{\z}_{k}}^{-1} -\R_{\tilde{\z}_{k}}^{-1}
%	K_{1}.
%\end{align}
\end{theorem}
\begin{proof}
	Please see Appendix~\ref{sec:proof-gradient-v}.
\end{proof}
Then, the next iteration point is calculated as $\bar{\v}^{(t+1)} = \exp \left(j\arg \left(\bar{\v}^{(t)} + \mu \nabla_{\bar{\v}^*} g(\bar{\v})\right)\right)$, where $\mu$ is the step size computed by the backtracking line search \cite{bubeck2015convex} and the $\arg$ operation is adopted for satisfying the unit-modulus constraint.

The overall algorithm to optimize $\bar{\v}$ is summarized as Algorithm 2. 
%The convergence of the algorithm to a stationary point solution of Problem (\ref{subproblem-v}) can be guaranteed since 
The objective values of Problem (\ref{subproblem-v}) are non-decreasing since the search direction is set as the steepest ascent direction $\nabla_{\bar{\v}^*} g(\bar{\v})$ and the backtracking line search is employed to find a suitable step size \cite{bubeck2015convex}.
In addition, since the solution set for $\bar{\v}$ is compact, the maximum value of WSKR is bounded. Therefore, the objective value converges over iterations.
%where the backtracking line search is used to find a suitable
%step size
%it is the channel gain is finite and the transmit power is limited, 
\begin{algorithm}[h]
	\caption{The PGA Algorithm for Problem (\ref{subproblem-v}).}
	\label{alg:2}
	\begin{algorithmic}[1]
		\Require $\R_{ i,j }^d, \R_{ i,j }^r,i,j=1,\cdots,K, \epsilon$.
		\State 
		Set: $t = 0$ (iteration index).
		\State Initial:  $\bar{\v}^{(0)}$.	
		\Repeat  
		\State Calculate the conjugate  gradient $\nabla_{\bar{\v}^*} g(\bar{\v}^{(t)})$ by Theorem \ref{lemma-v}.
		\State Find step size $\mu$ by backtrack line search \cite{bubeck2015convex}.
		\State Update ${\bar{\v}}^{(t+1)}$ as
		\begin{align}
			{\bar{\v}}^{(t+1)} =\exp\left(   j \arg\left({\bar{\v}}^{(t)} + \mu \nabla_{\bar{\v}^*} g(\bar{\v}^{(t)})\right)\right)  .
		\end{align}
		%			\State Use bisection search method to update $\lambda_{i}$;
%		\State Project ${\bar{\v}}^{(t+1)} $ onto unit modulus: ${\bar{\v}}^{(t+1)}  = \exp\left(   j \arg\left({\bar{\v}}^{(t+1)}\right)\right) $;
		\State Set $t = t+1$ and calculate the secret key rate $\tilde{g}(\bar{\v}^{(t+1)})$.
		\Until{$\frac{|\tilde{g}(\bar{\v}^{(t+1)}) -  \tilde{g}(\bar{\v}^{(t)})|}{\tilde{g}(\bar{\v}^{(t)})} \leq \epsilon$.}
	\end{algorithmic}
\end{algorithm}

%\subsection{Convergence Analysis}

\subsection{Complexity Analysis}
We analyze the computational complexity of the proposed algorithm. 
We assume the number of the AO iteration is $T_{\rm AO}$ and then calculate the complexity required to solve each sub-problem.

For Algorithm 1, in each iteration, the complexity to update $\P$ is $\mathcal{O}(K(M_eM)^3)$. 
%since given $\X \in \mathbb{C}^{m\times n}$ and $\Y \in \mathbb{C}^{m\times p}$, the complexity for computing $\X \Y$  is $\mathcal{O}(mnp)$ and the complexity for calculating the matrix inversion of matrix $\Z \in \mathbb{C}^{m \times m}$ is $\mathcal{O}(m^3)$. 
The complexity of evaluating the Lagrangian multipliers $\lambda_{i},i\in \{1,\cdots,K\}$ can be ignored. Hence, the complexity for Algorithm 1 is $\mathcal{O}(T_{\rm KKT} K(M_eM)^3)$, where $T_{\rm KKT}$ is the required number of iterations.
For Algorithm 2,  the optimization of the RIS phase shifts depends on the number of gradient updates, $T_{\rm PGA}$, and the amount of operations performed
in each gradient update.  The complexity for computing the gradient $\nabla_{\bar{\v}^*} g(\bar{\v})$ is $\mathcal{O}((MNL)^3)$. Thus, the complexity of Algorithm 2 is $\mathcal{O}(T_{\rm PGA}(MNL)^3)$.
 Thus, the overall  complexity of the proposed algorithm is $ \mathcal{O}(T_{\rm AO}( T_{\rm KKT}K M_e^3 M^3 + T_{\rm PGA}M^3N^3L^3))  $.
%\blue{We assume $\h \sim \mathcal{CN}(0, \R )$, then we have 
%	\begin{align}
%		\h = \R^{\frac{1}{2}} \tilde{\h} = \R^{\frac{1}{2}}
%	\frac{1}{\sqrt{N}}\tilde{\H} \v =\frac{1}{\sqrt{N}}\left(\v^T \otimes \I\right) \left(\I \otimes \R^{\frac{1}{2}}\right) \text{vec}\left(
%	\tilde{\H}\right)
%	= \left(\v^T \otimes \I\right)  \bar{\h},
%\end{align}
%where $\bar{\h} \sim  \mathcal{CN}\left(0, \frac{1}{\sqrt{N}}\left(\I \otimes \R^{\frac{1}{2}}\right)\right).
% $
%Thus, $\tilde{\z}_{k}$ can be expressed as
%\begin{align}
%	\tilde{\z}_{k}&=
%	\sum_{j=1}^{K}\left(\h_{k,j}^{d} + \H_{k,j}^{r}\v\right)+ \tilde{\n}^{u}_{k} \\
%	& = \sum_{j=1}^{K}\left(\left(\v^T \otimes \I\right)\bar{\h}_{k,j}^{d} + \left(\v^T \otimes \I\right)
%	\text{vec}\left(\H_{k,j}^{r}\right)\right)+ \left(\v^T \otimes \I\right)
%	\text{vec}
%	\left(\tilde{\N}^{u}_{k}\right).
%\end{align}
%}

\section{Simulation Results}
In this section, simulation results are presented to illustrate
the performance of the proposed multi-cell RIS-aided PKG scheme.

\subsection{Simulation Setup}
In the simulation, the channels are modeled by the product of large-scale path loss and  small-scale fading. In particular, the large-scale path loss is denoted as $ \sqrt{\zeta_{0} d^{-\alpha} } $, where $d$, $\zeta_{0}$, and $\alpha$ are the distance, path loss at 1 m, and the path loss exponent, respectively. 
Due to the extensive obstacles and scatters between the BSs and the UTs, the path loss exponents of the BS-UT (BU) links, BS-RIS (BR) links, and the RIS-UT (RU) links  are given by $\alpha_{\rm BU} = 3.75$ and $\alpha_{\rm BR} = \alpha_{\rm RU}  = \alpha_{\rm RIS}  = 2.2$, respectively \cite{Pancell}. 
The heights of the BSs, RISs, and UTs are $30$ m, $10$ m, and  $1.5$ m, respectively.
Furthermore, the small-scale fading between the BSs and UTs is assumed to follow Rayleigh fading and the small-scale fading  of the BS-RIS link and RIS-UT link are modeled as Rician distribution \cite{Yu_JSAC}. $\beta$ denotes the Rician factor and is set as $3$ in the simulation \cite{Pancell}.
%\begin{align}
%	\G_{i,l}^{r} = \sqrt{\frac{\beta_{i,l}}{1 + \beta_{i,l}}} \G_{i,l}^{r,{\rm LoS}} +
%	\sqrt{\frac{1}{1 + \beta_{i,l}}} \G_{i,l}^{r,{\rm NLoS}},
%\end{align}
%\begin{align}
%	\H = \sqrt{\frac{\beta}{1 + \beta}} \H^{{\rm LoS}} +
%	\sqrt{\frac{1}{1 + \beta}} \H^{{\rm NLoS}},
%\end{align}
%where $\beta$ is the Rician factor. $\H^{{\rm LoS}}$ and $\H^{{\rm NLoS}}$ are the line-of-sight (LoS) and non-LoS components, which are modeled as  the product of the array response vectors of the transceivers and Rayleigh fading, respectively. 
The noise power at the BSs and UTs is $-90$ dBm \cite{Yu_JSAC}. 
%Given the above settings, the channels between the BSs and UTs are full rank, thus the dimension of the precoding matrices $\P_{i}$ is set to $M_e = M = 4$.

Additionally, Two baseline schemes are adopted: (1) \textbf{No-RIS}: the RIS-related channels are set to zero and the precoding matrices at the BSs are optimized by Algorithm 1, and (2) \textbf{RandPhase}: the phase shifts of the RISs are random and the precoding matrices at the BSs are optimized by Algorithm 1.

\subsection{Two-Cell Scenario}
We first consider a two-cell case, as shown in Fig. \ref{fig:simulation_model}, where the $x$ and $y$ axes represent the horizontal plane and the $z$ axis represents the corresponding height.
The two BSs are located at $(0,0,10)$ and $(600,0,0)$, respectively, while there are two UTs at $(280,0,0)$ and $(320,0,0)$, respectively \cite{Pancell}. There is one RIS deployed at $(300, 0,0)$, which is the cell boundary.

%To provide more insights about the benefits of deploying RIS, we first consider the PKG in a two-cell system with a single RIS from Sec. VII-B to Sec. VII-G. Then, in Sec. VII-H, we study the optimal deployment of RISs and the performance comparison of single RIS and multiple RISs in a four-cell system.

%The results in the following subsections were averaged over $1000$ random channel realizations.
%\subsection{Two-Cell Scenario}
\begin{figure}
	\centering
	{\includegraphics[width=0.53\textwidth]{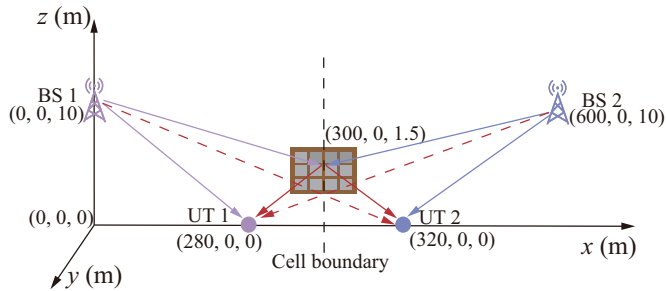}} %BDR/BDR
	\caption{Simulation setup for the RIS-assisted PKG in a two-cell scenario.}\label{fig:simulation_model}
\end{figure}

%\subsection{Convergence Behavior}
\begin{figure}
	\centering
	{\includegraphics[width=0.8\textwidth]{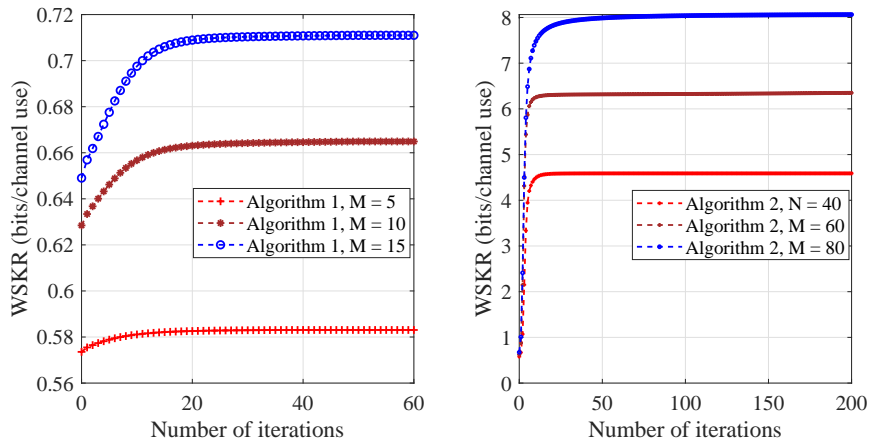}} %BDR/BDR
	\caption{The convergence of Algorithm 1 (left half of the figure) and Algorithm 2 (right half of the figure) for $P_{\rm A} = 30$ dBm.}\label{fig:convergence_inner}
\end{figure}

\begin{figure}
	\centering
	{\includegraphics[width=0.53\textwidth]{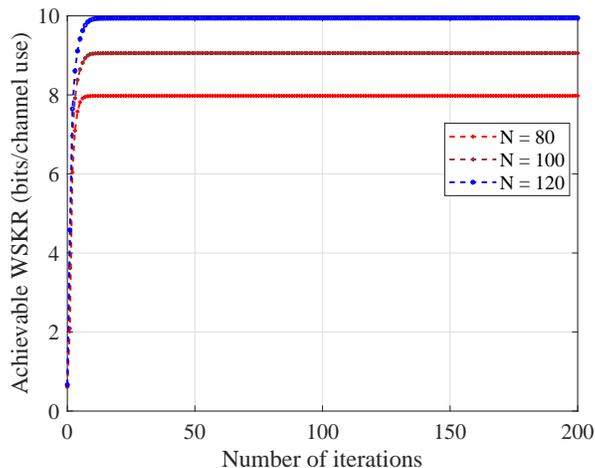}} %BDR/BDR
	\caption{The convergence of the AO algorithm for $M=M_e=4$, $P_{\rm A} = 30 $ dBm.}\label{fig:convergence_outer}
\end{figure}
%First, we evaluate the convergence behavior of the proposed algorithm. 
First, Fig. \ref{fig:convergence_inner} shows the convergence behavior of Algorithm 1 and Algorithm 2, which correspond to the left and right subfigures, respectively. It can be seen that Algorithm 1 and Algorithm 2 converge monotonically for all of the considered values of $M$ and $N$. This verifies the excellent convergence properties of Algorithms 1 and 2. Furthermore, Fig. \ref{fig:convergence_outer} depicts the achieved WSKR versus the number of AO iterations. As can be observed, the WSKR converges rapidly to stationary values after a few iterations on average. For $N=80$, the proposed algorithm converges after around $19$ iterations on average. For the case with more RIS elements, i.e., $N=100$, the number of iterations required for the convergence is $25$. This is because the solution space is enlarged when more optimization variables are involved and thus more iterations are required for convergence.

%\subsection{Impact of the Transmit Power}
%\begin{figure}
%	\centering
%	{\includegraphics[width=0.5\textwidth]{fig/simulation/R_P/R_P3.eps}} %BDR/BDR
%	\caption{Achievable WSKR versus the transmit power $P_{\rm A}$.}\label{fig:R_P}
%\end{figure}
\begin{figure}
	\centering
	{\includegraphics[width=0.53\textwidth]{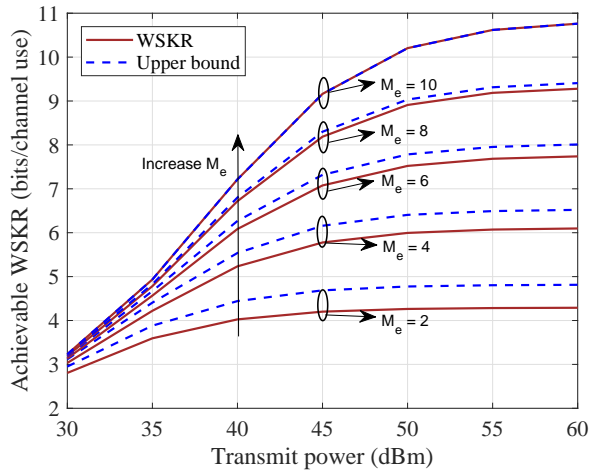}} %BDR/BDR
	\caption{Achievable WSKR versus the transmit power at the BSs, $P_{\rm A}$, for $M=10$ and $N=20$.}\label{fig:R_P}
\end{figure}
Next, in Fig. \ref{fig:R_P}, the WSKR versus the maximum transmit power at each BS for different $M_e$ is plotted. First, it can be observed that the WSKRs in all of the settings increase with the transmit power, since the impact of noise becomes insignificant. At the same time, the WSKR  increases with diminishing returns at high transmit power region. This is because only the transmit power at the BSs is increased while the power of the UTs is fixed that remains the system performance bottleneck. In particular, the BSs' channel estimations are still impaired by the noise components, causing the saturation in the WSKR. Furthermore, it can be seen that the WSKR increases with $M_e$, since a higher dimensional channel features can be exploited with a larger $M_e$. 
Finally, Fig. \ref{fig:R_P} shows the tightness of the derived upper bound. Specifically, the performance gap between WSKR and the upper bound is reduced with $M_e$ increasing and becomes exact when $M_e = M$. This is because having a larger $M_e$, the dimension of $\P_{k}\z_{k}$ approaches that of  $\z_{k}$. Thus, the upper bound $\mathcal{I}(\y_{k};\P_{k}\z_{k})$ becomes tight and accurate in describing $\mathcal{I}(\y_{k};\z_{k}) $ such that the proposed design is effective. 
%This observation indicates the proposed algorithm has the near-optimal performance.

%Consequently, the noises in BSs' channel estimation are not reduced, causing the WSKR to converge. 
%In addition, it is seen that the derived upper bound is always the same as the WSKR. This is because the optimized precoding matrices $\P_{i}$ is always full of rank. Specifically, according to (\ref{KKT}), the rank of $\P_{i}^{(t+1)}$ is equal to the rank of $\P_{i}^{(t)}$ since the channels between BSs and UTs are full of rank. Also, we set the initial value of $\P_{i}$ is $ \P_{i}^{(0)}= \sqrt{\frac{P_{\rm A}}{{M}} } \I_{M}$. As a result, the rank of $\P_{i}^{\rm KKT}$ is $M$, which satisfies the sufficient condition for $R_k^{ub} = R_k$ in Remark 4.
%Also, it is noted that with the number of RIS elements grows, the WSKR increases significantly. This is because RIS equipped more elements can alleviate more interference.
%Recalling Theorem , we have
%\begin{align}
%	&\sum_{k=1}^{K}w_{k}
%	\left(\left(\sum_{i=1}^{K}\P_{i}^{(t)}\M_{i,k}(\P_{i}^{(t)})^H + \I_{M} \right)^{-1}\P_{i}^{(t)}\M_{i,k}  \right) \notag \\
%	= &\lambda_{i}
%	\P_{i}^{(t+1)} +\sum_{k=1}^{K}w_{k}\left(\left(\I_{M} + 
%	\sum_{i=1}^{K}\P_{i}^{(t)}\M_{i,k}(\P_{i}^{(t)})^H -\P_{k}^{(t)}\N_{k,k}(\P_{k}^{(t)}) ^H \right)^{-1} \P_{i}^{(t+1)} \L_{i,k} \right),
%\end{align}
%which means when the rank of $\P_{i}^{(t)}$ is $N$, the rank of $\P_{i}^{(t+1)}$ is also $M$. In simulation, we set $\P_{i}^{(0)} = \sqrt{\frac{P_{\rm A}}{M}} \I_M $, and thus the rank of optimized $\P_{i}$ is $M$.

\begin{figure}
	\centering
	{\includegraphics[width=0.53\textwidth]{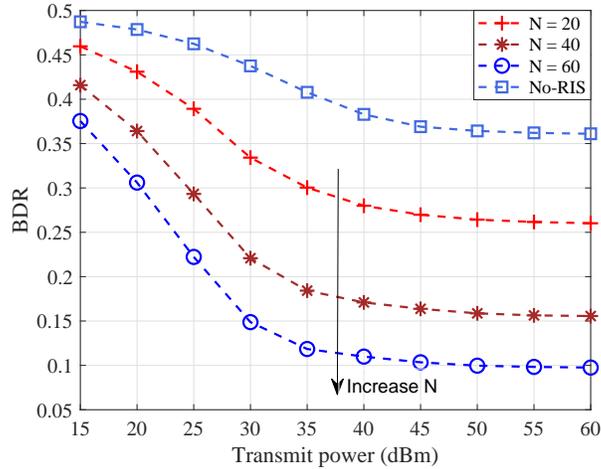}} %BDR/BDR
	\caption{BDR versus the transmit power $P_{\rm A}$ for $M=M_e=4$.}\label{fig:BDR_P}
\end{figure}
To illustrate the impact of the proposed algorithm on the channel reciprocity, we show the  BDR result in Fig. \ref{fig:BDR_P}. 
The BDR is defined as the ratio between the
number of disagreement bits and the number of total bits \cite{ZhangReview,Guillaume2015Bringing,2021Sum,mathur2008radio}.
%Firstly, it can be observed that all of the curves decrease with the increasing transmit power in low power region since the negative impact of the thernoise is reduced. 
%Besides, because the noise exists at BSs' channel estimation . 
%Then, 
It can be seen that the BDR of No-RIS baseline scheme is much higher than the proposed RIS-aided scheme and the former saturates at around $0.36$ BDR in the high transmit power regime. This is because the non-reciprocal interference significantly degrades the similarity of the uplink and downlink channel estimations.
In contrast, for the proposed RIS-aided PKG scheme, a lower BDR can be achieved with the increases of RIS elements number $N$, since the RIS is equipped with a higher ability to suppress the pilot contamination-caused inter-cell interference with larger $N$.

%\subsection{Impact of the Number of RIS elements}
\begin{figure}
	\centering
	{\includegraphics[width=0.53\textwidth]{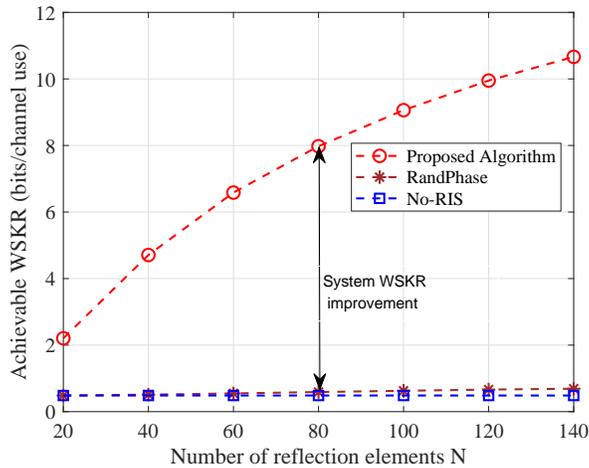}} %BDR/BDR
	\caption{Achievable WSKR versus the number of RIS elements $N$ for $P_{\rm A} = 30$ dBm, $M=M_e=4$.}\label{fig:R_N}
\end{figure}
Fig. \ref{fig:R_N} shows the WSKR of different PKG schemes versus the number of RIS elements $N$. Firstly, as can be observed, the WSKRs of the two baseline schemes are significantly low and approach zero. This is because the UTs are located at the cell edge and suffer from severe inter-cell pilot contamination. 
%In addition, the RandPhase scheme can only slightly improve the WSKR, since the introduce of RIS increases transmit power since the negative impacts of noises in the channel estimation are reduced. 
In contrast, the WSKR of the proposed RIS-based scheme is improved significantly the performance gain increases with the number of RIS elements. 
%Specifically, the WSKR is approximately $12$ bits/channel use for $N=140$, while the WSKRs for the baseline RandPhase and No-RIS schemes are only $0.49$ and $0.79$ bits/channel use, respectively.
 This is because the multi-cell interference in the uplink and downlink is eliminated by adjusting the phase shifts of the RIS and with more RIS elements in place, the proposed design
 becomes more flexible to create pencil-like beams to focus the reflected signals on the target UTs and BSs.

%\subsection{Impact of the RIS Location}
\begin{figure}
	\centering
	{\includegraphics[width=0.53\textwidth]{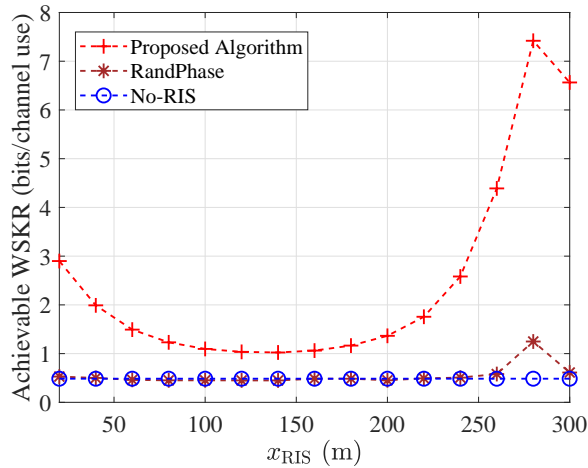}} %BDR/BDR
	\caption{Achievable WSKR versus the location of the RIS $x_{\rm RIS}$ for $M=M_e=4$, $N=60$, $P_{\rm A} = 30$ dBm.}\label{fig:R_RIS_location}
\end{figure}
%\subsection{Impact of the Transmit Power}
Fig. \ref{fig:R_RIS_location} presents the WSKR versus the coordinate of the RIS $(x_{\rm RIS},0,0)$. We vary the location of the RIS from near the cell center, $x_{\rm RIS} = 20$ m, to the cell boundary, $x_{\rm RIS} = 300$ m. It can be observed that the RandPhase baseline scheme achieves its maximum value when the RIS is located near one of the UTs, i.e., $x_{\rm RIS} = 280$, which is caused by the array gain brought by the RIS.
Also, the proposed algorithm enjoys superior performance than that of the two baseline schemes in all of the considered cases. 
It is interesting to observe that the WSKR achieved by the proposed algorithm first decreases from $x_{\rm RIS} = 20$ to $x_{\rm RIS} = 140$, then increases from $x_{\rm RIS} = 140$ to $x_{\rm RIS} = 280$, and finally decreases from $x_{\rm RIS} = 280$ to $x_{\rm RIS} = 300$ m. This is mainly caused by the variations of the large-scale fading of the RIS channels. To be specific, the large-scale channel gain of the RIS channel is the product of the gain of the BS-RIS link and that of the RIS-UT link.
Thus, for BS $2$ and UT $2$, the large-scale fading of the RIS channels can be approximated by $\zeta_{0} \sqrt{\left((600 - x_{\rm RIS})(320-x_{\rm RIS})\right)^{-\alpha_{\rm RIS}}}$, which is negligible for the considered $0\leq x_{\rm RIS}\leq 300$. 
However, for BS 1 and UT 1, when $0\leq x_{\rm RIS}\leq 280$, the large-scale fading of the RIS channels can be approximated by $\zeta_{0} \sqrt{\left(x_{\rm RIS}(280-x_{\rm RIS})\right)^{-\alpha_{\rm RIS}}}$, which achieves its minimum value at $x_{\rm RIS} = 140$  and its maximum value at $x_{\rm RIS} = 0$ or $x_{\rm RIS} = 280$. This means that the channel gain is minimal when the RIS is located at the middle point between BS 1 and UT 1, while the gain is maximum when the RIS is located close either to BS 1 or UT 1.
Furthermore, when the RIS is located away from UT 1, i.e., $ 280 < x_{\rm RIS}\leq 300$, the large-scale fading of the RIS channels between BS $1$ and UT $1$ is given by $\zeta_{0} \sqrt{\left(x_{\rm RIS}(x_{\rm RIS} - 280)\right)^{-\alpha_{\rm RIS}}}$, which decreases with an increasing $x_{\rm RIS}$.
 Additionally, it can be observed from Fig. \ref{fig:R_RIS_location} that the WSKR of deploying the RIS close to the BS, i.e., $x_{\rm RIS}$ approaches $0$, is less than that of deploying near the UT, i.e., $x_{\rm RIS}$ approaches $280$. This is because deploying the RIS  close to the cell center is less effective in alleviating the inter-cell pilot contamination.
%can only facilitate the local cell to improve the received signal strength of the UT, while it can alleviate the inter-cell interference to improve the WSKR of all cells when it is deployed at the cell boundary. 

%\subsection{Impact of the User Location}
\begin{figure}
	\centering
	{\includegraphics[width=0.53\textwidth]{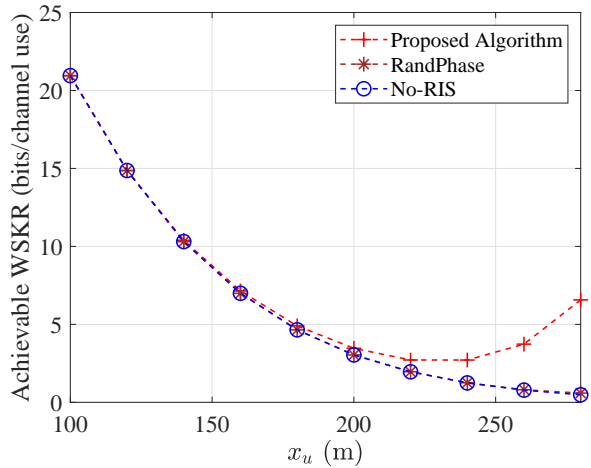}} %BDR/BDR
	\caption{Achievable WSKR versus the location of the RIS $x_{\rm UT}$ for $M=M_e=4$, $N=60$, $P_{\rm A} = 30$ dBm.}\label{fig:R_user_location}
\end{figure}
Fig. \ref{fig:R_user_location} shows the WSKR versus the location of the UTs under different schemes. The x-coordinates of UT 1 and UT 2 are $x_{u}$ and $600 - x_{u}$, respectively.
From this figure, the WSKR of all the schemes decrease with $x_{u}$ being from $100$ m to $280$ m.
This is mainly due to the following two reasons. 
First, the channel gain between the BSs and UTs decreases as the UTs move away from their home BSs.
Second, the inter-cell interference becomes dominated when the UTs move from the cell center to the cell boundary. 
However, the proposed RIS-aided PKG scheme provides significant performance gain when the UTs are close to the RIS. Specifically, when $x_{u} = 280$ m, the WSKR of the proposed scheme and the two baseline schemes are $6.57$ and $0.49$ bits/channel use, respectively. This is because the UTs receive strong reflected signals from the RIS and the inter-cell interference is substantially mitigated by the RIS.

%\subsection{Single RIS or Multiple RISs?}
\begin{figure}
	\centering
	{\includegraphics[width=0.53\textwidth]{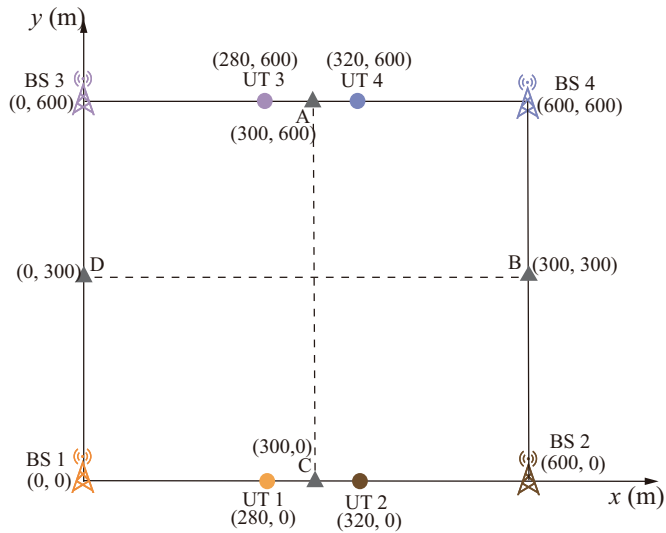}} %BDR/BDR
	\caption{Simulation setup for the RIS-aided PKG in a four-cell scenario. }\label{fig:Four-cell}
\end{figure}
\begin{figure}
	\centering
	{\includegraphics[width=0.53\textwidth]{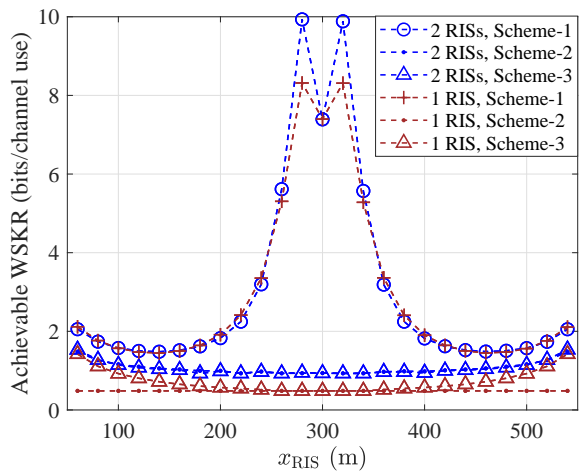}} %BDR/BDR
	\caption{Achievable WSKR versus the location of the RIS $x_{\rm RIS}$ in a four-cell scenario for $M=M_e=4$, $P_{\rm A} = 30$ dBm.}\label{fig:2RISs}
\end{figure}

\subsection{Four-Cell Scenario}
In this section, we consider a four-cell scenario to study the optimal RIS deployment and the impact of having single RIS and multiple RISs on the PKG performance. The coordinates of the four BSs are $(0,0)$, $(600,0)$, $(600,600)$, and $(0,600)$, respectively. The four UTs are located as $(280,0)$, $(320,0)$, $(280,600)$, and $(320,600)$, respectively.

For the single RIS case, we consider 3 schemes: \textbf{Scheme-1}: The RIS moves from BS 1 to BS 2;
\textbf{Scheme-2}: The RIS moves from point D to point B;
\textbf{Scheme-3}: The RIS moves from BS 3 to BS 2. 
For the two RISs case, we consider 3 schemes: \textbf{Scheme-1}: RIS 1 moves from BS 1 to BS 2 and RIS 2 moves from BS 3 to BS 4;
\textbf{Scheme-2}: RIS 1 moves from BS 1 to BS 4, and RIS 2 moves from BS 3 to BS 2;
\textbf{Scheme-3}: RIS 1 moves from point D to point B, and RIS 2 moves from point A to point C.

Fig. \ref{fig:2RISs} shows the WSKR versus the coordinate of the RIS under different schemes. 
Firstly, for the single-RIS scenario, it is seen that \textbf{Scheme-1} achieves the maximum WSKR at $x_{\rm RIS} = 280$ m and $x_{\rm RIS} = 320$ m, which are the locations of UT 1 and UT 2, respectively. This observation is consistent with that of the two-cell system.  Additionally, the WSKR of \textbf{Scheme-1} is significantly higher than that of \textbf{Scheme-2} and \textbf{Scheme-3}. The reason behind this is the channel gain of the RIS-related channel is higher in \textbf{Scheme-1} when the RIS is close to the two UTs.
%the RIS is more close to UT 1 and UT 2 while is far away from the UTs in cell 3 and cell 4. Therefore, UT 3 and UT 4 can get negligible performance gain. 
%This motivates us to deploy multiple RISs to facilitate more cells.
Secondly, for the two-RISs scenario, the WSKR curves are similar to that of single-RIS case and the WSKR of \textbf{Scheme-1} has the maximum value at $x_{\rm RIS} = 280$ m and $x_{\rm RIS} = 320$ m.
The reason is that the RISs are respectively closer to the two of the UTs at these two points.
By comparing the WSKR of single-RIS and two-RISs at $x_{\rm RIS} = 280$ and $x_{\rm RIS} = 320$, we can observe that two-RISs achieve higher performance gain than the single-RIS case. This is because when multiple
RISs are deployed in the network, the distance between each UT and its nearest RIS is reduced due to spatial diversity, which thus increases the channel gain. This observation confirms that the PKG performance of the UTs is usually dominated by the closest RIS.
% Therefore, the number of RISs, i.e., $L$, generally depends on the number of UTs that suffer from pilot contamination.
%the distributed RIS deployment has superior performance than the centralized deployment. 

%\subsection{Randomness Evaluation}
\begin{table}[htbp]
	\centering  
	\caption{NIST random test result}  
	%	\vspace{1em}
	\label{table1}  
	\begin{tabular}{|c|c|c|c|c|c|c|c|c|c|}  
		\hline
		& Pass ratio & P-value\\
		\hline
		Approximate entropy & 0.9911 & 0.4851 \\ \hline
		Runs &  0.9985 & 0.5088 \\ \hline
		Ranking &  0.9941 & 0.4992 \\ \hline
		Longest runs of ones & 0.9899 & 0.6238 \\ \hline
		Frequency & 1 & 0.5201 \\ \hline
		FFT & 0.9853 & 0.4295 \\ \hline
		Block frequency & 1  & 0.4401 \\ \hline
		Cumulative sums & 0.9981 & 0.5057\\ \hline
		Serial & 0.9838 & 0.5013, 0.4895\\\hline
	\end{tabular}
\end{table} 
Finally, to verify the randomness of the obtained bit sequences for cryptographic applications, we conduct the National Institute of Standards and Technology (NIST) randomness test~\cite{rukhin2000} on the quantized bits. 
Note that NIST  is widely adopted to
evaluate the randomness of true-random and pseudo-random number generators. The output of the NIST is p-values, which is examined to ensure uniformity. The tested sequence is considered to  pass the test if the p-value is greater than $0.01$. In the simulation, we perform $9$ kinds of  NIST items for
$10,000$ trials, with the length of each sequence being $2048$ bits. The test results are shown in Table I, where the pass ratio represents the  number of passed trials over the $10,000$ trials. As can be observed, the p-values of all the test items are significantly greater than $0.01$, which means the sequence can be considered to be uniformly distributed. 
Also, the  pass ratios are all higher than 98\%. This confirms the excellent randomness of the  bits generated by the proposed  PKG scheme.

\section{Conclusion}
In this paper, we incorporated RISs in multi-cell PKG systems to alleviate the negative impact of multi-cell pilot contamination on PKG performance. Specifically, we studied the WSKR maximization problem by jointly optimizing the precoding matrices at the BSs and the phase shifts at the RISs.
%, while guaranteeing each BS’s power constraint and unit-modulus constraint of each RIS element. 
To tackle this non-convex problem, 
we derived a tight upper bound of the objective function. To solve the upper bound maximization problem, we applied an AO-based algorithm that alternatively solves the sub-problem for precoding matrices and the sub-problem for phase-shifting vector. In particular, a Lagrangian dual algorithm and a PGA algorithm were employed to design the precoding matrices and the phase-shifting vector, respectively.
Simulation results verified the WSKR of the cell edge UTs under multi-cell pilot contamination can be unsatisfactory when the RISs are not deployed. By contrast, the proposed RIS-aided PKG scheme can achieve high WSKR even in harsh channel conditions. 
Moreover, a lower BDR can be observed with the increase of RIS elements number.
%distributed RISs, with each RIS being located in the proximity of some UTs, offer rich spatial diversity to maximize the WSKR.
%it is found that deploying the RISs with each RIS close to the UT achieves the maximum WSKR.

\appendix

%% appendices
\begin{appendices}

	\subsection{Proof of Theorem~1}\label{sec:proof-KGR-RIS}
%	\begin{proof}
%		\begin{align}
%			R_{k} = \mathcal{I}(\y_{k};\z_{k})
%			%	=& \mathcal{H}\left(\hat{ h}_{a}, \hat{ h}_{e_{k}}^{d}\right)+\mathcal{H}\left(\hat{ h}_{b}, \hat{ h}_{e_{k}}^{d}\right) \notag \\
%			%		&-\mathcal{H}\left(\hat{ h}_{a}, \hat{ h}_{b}, \hat{ h}_{e_{k}}^{d}\right)-\mathcal{H}\left(\hat{ h}_{e_{k}}^{d}\right) \\
%			= \log_{2} \frac{\operatorname{det}\left(\R_{\y_{k}} \R_{\z_{k}}\right)}{\operatorname{det}\left(\mathcal{R}_{\y_{k}\z_{k}} \right)} , \label{eq:MI}
%		\end{align}
%		where the covariance matrices are denoted as $\R_{\y_{k}}=\mathbb{E}\left\{ \y_{k}\y_{k}^H \right\}$, $\R_{\z_{k}}=\mathbb{E}\left\{ \z_{k}\z_{k}^H \right\}$, and
%		\begin{align}
%			\mathcal{R}_{\y_{k}\z_{k}} =\left[\begin{array}{ll}
%				\R_{\y_{k}} & \R_{\y_{k}\z_{k}}  \\
%				\R_{\z_{k}\y_{k}} & \R_{\z_{k}}
%			\end{array}\right].
%		\end{align}
%		Employing the determinant of a block matrix, we can have 
%		\begin{align}
%			\det\left(\R_{\y_{k}\z_{k}} \right) = \det\left( \R_{\y_{k}} \right) \det\left( \R_{\z_{k}} - \R_{\z_{k}\y_{k}}\R_{\y_{k}}^{-1} \R_{\y_{k}\z_{k}}\right) ,
%		\end{align}
%		and the KGR in $k$-th cell is expressed as
%		\begin{align}
%			R_{k} &= \log_{2} \frac{\det\left(\R_{\y_{k}}\right)\det\left(\R_{\z_{k}}\right)}{\det\left(\R_{\y_{k}}\right) \det\left(\R_{\z_{k}} - \R_{\z_{k}\y_{k}}\R_{\y_{k}}^{-1} \R_{\y_{k}\z_{k}}\right)} \\
%			& =- \log_{2} \det\left(\I -  \R_{\z_{k}\y_{k}}\R_{\y_{k}}^{-1} \R_{\y_{k}\z_{k}} \R_{\z_{k}}^{-1} \right). \label{Rk0}
%		\end{align}
			The KGR between BS $k$ and UT $k$ is expressed as \cite{jorswieck2013secret}, \cite{wong2009secret}
		\begin{align}
			R_{k} = \mathcal{I}(\y_{k};\z_{k})
			%	=& \mathcal{H}\left(\hat{ h}_{a}, \hat{ h}_{e_{k}}^{d}\right)+\mathcal{H}\left(\hat{ h}_{b}, \hat{ h}_{e_{k}}^{d}\right) \notag \\
			%		&-\mathcal{H}\left(\hat{ h}_{a}, \hat{ h}_{b}, \hat{ h}_{e_{k}}^{d}\right)-\mathcal{H}\left(\hat{ h}_{e_{k}}^{d}\right) \\
			= \ln \frac{\operatorname{det}\left(\R_{\y_{k}} \R_{\z_{k}}\right)}{\operatorname{det}\left(\mathcal{R}_{\y_{k}\z_{k}} \right)} , \label{eq:MI0}
		\end{align}
		where the covariance matrices are denoted as $\R_{\y_{k}}=\mathbb{E}\left\{ \y_{k}\y_{k}^H \right\}$, $\R_{\z_{k}}=\mathbb{E}\left\{ \z_{k}\z_{k}^H \right\}$, and
		\begin{align}
			\mathcal{R}_{\y_{k}\z_{k}} =\left[\begin{array}{ll}
				\R_{\y_{k}} & \R_{\y_{k}\z_{k}}  \\
				\R_{\z_{k}\y_{k}} & \R_{\z_{k}}
			\end{array}\right].
		\end{align}
		Employing the determinant of a block matrix, we have 
		\begin{align}
			\det\left(\R_{\y_{k}\z_{k}} \right) = \det\left( \R_{\y_{k}} \right) \det\left( \R_{\z_{k}} - \R_{\z_{k}\y_{k}}\R_{\y_{k}}^{-1} \R_{\y_{k}\z_{k}}\right).
		\end{align}
		Thus, the KGR  is reformulated as
		\begin{align}
			R_{k} &= \ln \left(\frac{\det\left(\R_{\y_{k}}\right)\det\left(\R_{\z_{k}}\right)}{\det\left(\R_{\y_{k}}\right) \det\left(\R_{\z_{k}} - \R_{\z_{k}\y_{k}}\R_{\y_{k}}^{-1} \R_{\y_{k}\z_{k}}\right)}\right) \\
			& =- \ln \det\left(\I -  \R_{\z_{k}\y_{k}}\R_{\y_{k}}^{-1} \R_{\y_{k}\z_{k}} \R_{\z_{k}}^{-1} \right). \label{Rk00}
		\end{align}
		Given the channel estimations in (\ref{yk}) and (\ref{zk}),
%		as
%			\begin{align}
%				\z_{k} 
%				&=
%				\P_{k}\left[\sum_{j=1}^{K}\h_{k,j}^{d} +
%				\left(\bar{\v}^H \otimes \I_{M} \right)\sum_{j=1}^{K}
%				\text{vec} \left(\H_{k,j}^{r}\right)+ \tilde{\n}^{u}_{k}\right],\\
%				\y_{k} 
%				& = \sum_{i=1}^{K}\P_{i}\left(\h_{i,k}^{d} + \left(\bar{\v}^H \otimes \I_{M} \right)
%				\text{vec}
%				\left(\H_{i,k}^{r}\right)\right)  + \tilde{\n}_{k}^{d}
%				.
%			\end{align}
%		where $\P_{i} \in \mathbb{C}^{M \times M}$, $\h_{k,j} \in \mathbb{C}^{M \times 1}$, $\bar{\v} \in \mathbb{C}^{N \times 1}$, $\H_{k,j}^{r} \in \mathbb{C}^{M \times N}$, and $\tilde{\n}^{u}_{k} \in \mathbb{C}^{M \times 1}$.
the channel covariance matrices are calculated as 
		\begin{align}
			\R_{ {\z}_{k}}
			&=\P_{k}\left[\sum_{j=1}^{K}\R_{k,j}^{d} +
			\left(\bar{\v}^H \otimes \I_{M} \right)\sum_{j=1}^{K}
			\R_{k,j}^{r}\left(\bar{\v} \otimes \I_{M} \right)+\I_{M}\right]\P_{k}^H,\\
			\R_{\y_{k}} &= 
			\sum_{i=1}^{K}\P_{i}\left(\R_{i,k}^{d} + \left(\bar{\v}^H \otimes \I_{M}\right)
			\R_{i,k}^{r}\left(\bar{\v}^H \otimes \I_{M}\right)^H \right) \P_{i}^H +\I_{M_e},\\
			\R_{ \y_{k}{\z}_{k}}&=
			\P_{k}\left(\R_{k,k}^{d} + \left(\bar{\v}^H \otimes \I_{M}\right)
			\R_{k,k}^{r}\left(\bar{\v} \otimes \I_{M}\right)\right)\P_{k}^H = \R_{ {\z}_{k}\y_{k}} .
		\end{align}
		Substituting these channel covariance matrices into (\ref{Rk00}) and the result follows immediately.
%		we can obtain $R_k$ as (\ref{Th1}).
%		This completes the proof.
%	\end{proof}
	
	\subsection{Proof of Proposition~1} \label{KGR-analysis}
	First, the eigenvalue decomposition of $\R^d_{i,k}$ can be described as $\Q_{\rm D} \diag \left( \x \right) \Q_{\rm D}^H$, where $\x = [x_1,\cdots,x_M]$ is the eigenvalue vector and $\Q_{\rm D} = \left[\c_1, \cdots, \c_M\right]$ is the eigenmatrix and $\c_\ell, \ell \in \left\{1,\cdots,M\right\}$, is the  eigenvector. By denoting $ \sum_{j=1}^{K}\R^d_{k,j}+\I_{M} = \A$ and $ P_i\sum_{i=1}^{K}\R^d_{i,k}+\I_{M} = \B$,
	%		, and $\sum_{q=1,q\neq i}^{K} P_q \R_{ q,k } +\I_{M} = \Z$,  
	%		the $R_k$ can be written as
	%		\begin{align}
	%			R_k &= \ln \det \left(\A\right) - \ln \det \left[ \A - P_k\R_{k,k} \left(\Q_{\rm D} \diag \left(\x\right) \Q_{\rm D}^H + \B\right)^{-1} \R_{k,k}\right] \\
	%			%			& = \log \det \left(\A\right) - \log \det \left( \A - \C \diag \left(\x\right) \C^H\right) \\
	%			& = \ln \det \left(\A\right) - \ln \det \left[ \A - P_k\R_{k,k}(x_1 \c_1\c_1^H + \cdots + x_M \c_M\c_M^H + \B)^{-1} \R_{k,k} \right],
	%		\end{align}
	%		where  $\Q_{\rm D} = \left[\c_1, \cdots, \c_M\right]$. 
	%		Assuming $x_i \I = \X_i$, we can calculate the derivative of $R_k$ w.r.t. $\X_i$ as 
	%		\begin{align}
	%			 \partial R_k/\partial \X_i  = \left( \A - \left(\c_1 \c_1^H \X_1 + \cdots + \c_M \c_M^H \X_M\right) \right)^{-1} \c_i \c_i^H,
	%		\end{align}
	%		Thus, 
	the derivative of $R_k$ w.r.t. $x_{\ell},\ell \in \left\{1,\cdots,M\right\}$, is
	\begin{align}
		\partial R_k/\partial x_{\ell} = - \tr \left(
		\B ^{-1}
		\R^d_{k,k}
		\left( \A - P_k\R^d_{k,k} \B^{-1}\R^d_{k,k}
		\right)^{-1} 
		P_{k} \R^d_{k,k} 
		\B^{-1}
		P_i\c_{\ell} \c_{\ell}^H\right).
	\end{align}
	Since $ \A - P_k\R^d_{k,k} \B^{-1}\R^d_{k,k} \succ \mat{0}$ and $\c_{\ell} \c_{\ell}^H \succeq \mat{0}$, we have $\partial R_k/\partial x_{\ell}< 0$.
	%		since $\partial R_k/\partial \X_i 
	%		\succeq 0 $. 
	Thus, $R_k$ decreases monotonically	 with $x_{\ell}$. 
	%		Also, since $x_{\ell} =\frac{1}{\lambda_{\ell}\left( \sum_{i=1}^{K} \R_{i,k}  \right) + 1}$, $R_k$ decreases with $\lambda_{\ell}\left( \sum_{i=1}^{K} \R_{i,k}  \right)$. 
	When $\lambda_{\ell}\left( \R^d_{ i,k }\right) \rightarrow \infty$, we have $\lambda_{\ell}\left( \sum_{i=1}^{K}P_i\R^d_{ i,k } + \I_{M}\right) \rightarrow 0$, and thus
	$R_k \rightarrow 0$. The proof for $ \R^d_{k,j}$ and $P_i$ can be derived similarly.
	This completes the poof.

	\subsection{Proof of Theorem~2}\label{sec:proof-KGR-upperbound}
%	\begin{proof}
%		We first express the channel estimations as
%%		\blue{
%			\begin{align}
%				\tilde{\z}_{k} 
%%				&=  \sum_{j=1}^{K}\h_{k,j}  + \tilde{\n}^{u}_{k} \\
%%				&=
%%				\sum_{j=1}^{K}\left(\h_{k,j}^{d} + \H_{k,j}^{r}\v\right)+ \tilde{\n}^{u}_{k}\\
%				& =
%				\sum_{j=1}^{K}\h_{k,j}^{d} +
%				\left(\bar{\v}^H \otimes \I_{M} \right)\sum_{j=1}^{K}
%				\text{vec} \left(\H_{k,j}^{r}\right)+ \tilde{\n}^{u}_{k},
%			\end{align}
%			\begin{align}
%				\y_{k} 
%%				&= \sum_{i=1}^{K}\P_{i}\h_{i,k}  + \tilde{\n}_{k}^{d} \\
%%				&=\sum_{i=1}^{K}\P_{i}\left(\h_{i,k}^{d} + \H_{i,k}^{r}\v\right)  + \tilde{\n}_{k}^{d} \\
%				& = \sum_{i=1}^{K}\P_{i}\left(\h_{i,k}^{d} + \left(\bar{\v}^H \otimes \I_{M} \right)
%				\text{vec}
%				\left(\H_{i,k}^{r}\right)\right)  + \tilde{\n}_{k}^{d}
%				,
%			\end{align}
%		}
%		where $\P_{i} \in \mathbb{C}^{M \times M}$, $\h_{k,j} \in \mathbb{C}^{M \times 1}$, $\bar{\v} \in \mathbb{C}^{N \times 1}$, $\H_{k,j}^{r} \in \mathbb{C}^{M \times N}$, and $\tilde{\n}^{u}_{k} \in \mathbb{C}^{M \times 1}$.
		According to the data-processing inequality, we have $\mathcal{I}(\y_{k};\z_{k}) \leq \mathcal{I}(\y_{k};\tilde{\z}_{k})$.
%		\begin{align}
%		\end{align}
		Next, we calculate the mutual information $I(\y_{k};\tilde{\z}_{k})$.
		The channel covariance matrices are calculated as
		\begin{align}
			\R_{ \tilde{\z}_{k}}
			&=\sum_{j=1}^{K}\R_{k,j}^{d} +
			\left(\bar{\v}^H \otimes \I_{M} \right)\sum_{j=1}^{K}
			\R_{k,j}^{r}\left(\bar{\v} \otimes \I_{M} \right)+\I_{M},\\
			\R_{\y_{k}} &= 
			\sum_{i=1}^{K}\P_{i}\left(\R_{i,k}^{d} + \left(\bar{\v}^H \otimes \I_{M}\right)
			\R_{i,k}^{r}\left(\bar{\v}^H \otimes \I_{M}\right)^H \right) \P_{i}^H +\I_{M_e},\\
			\R_{ \y_{k}\tilde{\z}_{k}} &=
			\P_{k}\left(\R_{k,k}^{d} + \left(\bar{\v}^H \otimes \I_{M}\right)
			\R_{k,k}^{r}\left(\bar{\v} \otimes \I_{M}\right)\right) , \\
			\R_{ \tilde{\z}_{k}\y_{k}} &=
			\left(\R_{k,k}^{d} + \left(\bar{\v}^H \otimes \I_{M}\right)
			\R_{k,k}^{r}\left(\bar{\v} \otimes \I_{M}\right)\right)\P_{k}^H .
		\end{align}
		Substituting these channel covariance matrices into (\ref{Rk00}), the upper bound of the KGR is expressed as (\ref{KGR:ub}).
%		\begin{align}
%			R_{k}^{ub}
%			&=
%			\log \det \left(\R_{\y_{k}}\right) - 
%			\log \det\left(\R_{\y_{k}} -  \R_{\y_{k}\tilde{\z}_{k}}\R_{\tilde{\z}_{k}}^{-1} \R_{\tilde{\z}_{k}\y_{k}}  \right) \\
%			& = \log \det \left(\sum_{i=1}^{K}\P_{i}\M_{i,k} \P_{i}^H +\I_{M}\right) - 
%			\log \det\left(\sum_{i=1}^{K}\P_{i}\M_{i,k} \P_{i}^H +\I_{M} -  
%			\P_{k} \N_{k,k} \P_{k}^H  \right),
%		\end{align}
%		where 
%		\begin{align}
%			\M_{i,k} &= \R_{i,k}^{d} + \left(\bar{\v}^H \otimes \I_{M}\right)
%			\R_{i,k}^{r}\left(\bar{\v}^H \otimes \I_{M}\right)^H \\
%			\N_{k,k} & = \left(\R_{k,k}^{d} + \left(\bar{\v}^H \otimes \I_{M}\right)
%			\R_{k,k}^{r}\left(\bar{\v} \otimes \I_{M}\right)\right)
%			\R_{\tilde{\z}_{k}}^{-1} \left(\R_{k,k}^{d} + \left(\bar{\v}^H \otimes \I_{M}\right)
%			\R_{k,k}^{r}\left(\bar{\v} \otimes \I_{M}\right)\right).
%		\end{align}
		This completes the proof.
%	\end{proof}
	
	\subsection{Proof of Theorem~3}\label{sec:proof-KKT}
		We can calculate the conjugate gradient of $\mathcal{L}\left(\P,{\mat  \Lambda}\right)$ w.r.t. $\P$ as
		\begin{align}
			\frac{\partial \mathcal{L}\left(\P,{\mat  \Lambda}\right)}{\partial \P^{*}}  
			&=
				\sum_{k=1}^{K}w_{k} \left( \left(\P \M_{k} \P^H +\I_{M_e}\right)^{-1} \P \M_{k}  - 
				\left(\P \N_{k} \P^H +\I_{M_e}\right)^{-1} \P \N_{k} 
				\right)  -  \P {\mat  \Lambda} .
		\end{align}
		Set $\frac{\partial \mathcal{L}\left(\P,{\mat \Lambda}\right)}{\partial \P^{*}}  = \mat{0}$ and take the $\text{vec}(\cdot)$ operation on the both sides of the equation, we can obtain the KKT condition as shown in (\ref{KKT}). 	This completes the proof.

	\subsection{Proof of Theorem~4}\label{sec:proof-gradient-v}
%	\begin{proof}
		First, the upper bound is given by
		\begin{align}
			R_{k}^{\rm ub} 
			&=
			\ln \det \left(\R_{\y_{k}}\right) - 
			\ln \det\left(\R_{\y_{k}} -  \R_{\y_{k}\tilde{\z}_{k}}\R_{\tilde{\z}_{k}}^{-1} \R_{\tilde{\z}_{k}\y_{k}}  \right) 
		\end{align}
		and the differential of $R_{k}^{\rm ub} $ w.r.t. $\bar{\v}$ is written as
		\begin{align}
			d \left(R_{k}^{\rm ub} \right) & = \tr \left(
			\R_{\y_{k}}^{-1} d\left(\R_{\y_{k}}\right)
			\right)
			- \tr \left(
			\left(\R_{\y_{k}} -  \R_{\y_{k}\tilde{\z}_{k}}\R_{\tilde{\z}_{k}}^{-1} \R_{\tilde{\z}_{k}\y_{k}} \right)^{-1} d\left(\R_{\y_{k}} -  \R_{\y_{k}\tilde{\z}_{k}}\R_{\tilde{\z}_{k}}^{-1} \R_{\tilde{\z}_{k}\y_{k}} \right)
			\right). \label{differ}
		\end{align}
		Then, the first term of (\ref{differ}) is calculated as
		\begin{align}
			\tr \left(
			\R_{\y_{k}}^{-1} d\left(\R_{\y_{k}}\right)
			\right) 
%			&=
%			\tr \left(
%			\R_{\y_{k}}^{-1} d\left({\left(\bar{\v}^H \otimes \I_{M} \right)}
%			\sum_{i=1}^{K}
%			\left(\I_{N} \otimes \P_{i} \right)
%			\R_{i,k}^{r}\left(\I_{N} \otimes \P_{i} \right)^H
%			{\left(\bar{\v} \otimes \I_{M} \right)}\right)
%			\right) \\
%			& = 
%			\tr \left(
%			\sum_{i=1}^{K}
%			\left(\I_{N} \otimes \P_{i} \right)
%			\R_{i,k}^{r}\left(\I_{N} \otimes \P_{i} \right)^H
%			\left(\bar{\v} \otimes \I_{M} \right)
%			\R_{\y_{k}}^{-1} d\blue{\left(\bar{\v}^H \otimes \I_{M} \right)}
%			\right)\\
%			& = \tr \left(
%			\sum_{i=1}^{K}
%			\left(\I_{N} \otimes \P_{i} \right)
%			\R_{i,k}^{r}\left(\I_{N} \otimes \P_{i} \right)^H
%			\left(\bar{\v} \otimes \I_{M} \right)
%			\R_{\y_{k}}^{-1} 
%			K_{1}
%			d\blue{\left(\I \otimes \bar{\v}^H \right)} K_{2}
%			\right)  \\
			& = \tr \left({\mat{K}_{M_eNL}
				\sum_{i=1}^{K}
				\left(\I_{NL} \otimes \P_{i} \right)
				\R_{i,k}^{r}\left(\I_{NL} \otimes \P_{i} \right)^H
				\left(\bar{\v} \otimes \I_{M_e} \right)
				\R_{\y_{k}}^{-1} 
				\mat{K}_{M_e}}
			{\left(\I_{M_e} \otimes d\left(\bar{\v}^H\right) \right)} 
			\right) . 
%			\\
%			& = \tr \left(\Q_{k}
%			{\left(\I \otimes d\left(\bar{\v}^H\right) \right)} 
%			\right) =  \tr 
%			\left( \left(\sum_{m=1}^{M}\bar{\q}_{k,m}\right) d\left(\bar{\v}^H \right)
%			\right).
		\end{align}
%		Thus, 
%		\begin{align}
%			\tr \left(
%			\Q_{k}
%			\left( \I  \otimes d\left(\bar{\v}^H \right)
%			\right)
%			\right) = \sum_{m=1}^{M} \tr 
%			\left( \bar{\Q}_{k,m} d\left(\bar{\v}^H \right)
%			\right) 
%			=  \tr 
%			\left( \left(\sum_{m=1}^{M}\bar{\Q}_{k,m}\right) d\left(\bar{\v}^H \right)
%			\right).
%		\end{align}

		Similarly,  we can calculate the second term of (\ref{differ})  as% -------------
		\begin{align}
			&\quad \tr \left(
			\left(\R_{\y_{k}} -  \R_{\y_{k}\tilde{\z}_{k}}\R_{\tilde{\z}_{k}}^{-1} \R_{\tilde{\z}_{k}\y_{k}} \right)^{-1} d\left(\R_{\y_{k}} -  \R_{\y_{k}\tilde{\z}_{k}}\R_{\tilde{\z}_{k}}^{-1} \R_{\tilde{\z}_{k}\y_{k}} \right)
			\right) \\
%			&= 
%			\tr \left(
%			\left(\R_{\y_{k}} -  \R_{\y_{k}\tilde{\z}_{k}}\R_{\tilde{\z}_{k}}^{-1} \R_{\tilde{\z}_{k}\y_{k}} \right)^{-1} d\left(\R_{\y_{k}} \right)
%			\right)- 
%			\tr \left(
%			\left(\R_{\y_{k}} -  \R_{\y_{k}\tilde{\z}_{k}}\R_{\tilde{\z}_{k}}^{-1} \R_{\tilde{\z}_{k}\y_{k}} \right)^{-1} d\left(\R_{\y_{k}\tilde{\z}_{k}}\R_{\tilde{\z}_{k}}^{-1} \R_{\tilde{\z}_{k}\y_{k}}  \right)
%			\right)\\
%			& = 
%			\tr \left(
%			\sum_{i=1}^{K}
%			\left(\I \otimes \P_{i} \right)
%			\R_{i,k}^{r}\left(\I \otimes \P_{i} \right)^H
%			\left(\bar{\v} \otimes \I \right)
%			\left(\R_{\y_{k}} -  \R_{\y_{k}\tilde{\z}_{k}}\R_{\tilde{\z}_{k}}^{-1} \R_{\tilde{\z}_{k}\y_{k}} \right)^{-1} \blue{d\left(\bar{\v}^H \otimes \I \right)}
%			\right)  \notag \\
%			& \quad -\tr \left(
%			\left(\R_{\y_{k}} -  \R_{\y_{k}\tilde{\z}_{k}}\R_{\tilde{\z}_{k}}^{-1} \R_{\tilde{\z}_{k}\y_{k}} \right)^{-1} \blue{\left(
%				d\left(\R_{\y_{k}\tilde{\z}_{k}}\right)\R_{\tilde{\z}_{k}}^{-1} \R_{\tilde{\z}_{k}\y_{k}}  
%				+ \R_{\y_{k}\tilde{\z}_{k}}d\left(\R_{\tilde{\z}_{k}}^{-1}\right) \R_{\tilde{\z}_{k}\y_{k}}  
%				+ \R_{\y_{k}\tilde{\z}_{k}}\R_{\tilde{\z}_{k}}^{-1} d\left(\R_{\tilde{\z}_{k}\y_{k}}  \right)
%				\right)}
%			\right)\\
			& = 
			\tr \left({\mat{K}_{M_eNL}
				\sum_{i=1}^{K}
				\left(\I_{NL} \otimes \P_{i} \right)
				\R_{i,k}^{r}\left(\I_{NL} \otimes \P_{i} \right)^H
				\left(\bar{\v} \otimes \I_{M_e} \right)
				\left(\R_{\y_{k}} -  \R_{\y_{k}\tilde{\z}_{k}}\R_{\tilde{\z}_{k}}^{-1} \R_{\tilde{\z}_{k}\y_{k}} \right)^{-1} \mat{K}_{M_e}}{\left( \I_{M_e} \otimes d\left(\bar{\v}^H\right) \right)} 
			\right)  \notag \\
			& \quad -\tr \left(\R_{\tilde{\z}_{k}}^{-1} \R_{\tilde{\z}_{k}\y_{k}} 
			\left(\R_{\y_{k}} -  \R_{\y_{k}\tilde{\z}_{k}}\R_{\tilde{\z}_{k}}^{-1} \R_{\tilde{\z}_{k}\y_{k}} \right)^{-1} 
			{
				d\left(\R_{\y_{k}\tilde{\z}_{k}}\right)
			}
			\right)  \notag \\
			& \quad -\tr \left(\R_{\tilde{\z}_{k}\y_{k}}  
			\left(\R_{\y_{k}} -  \R_{\y_{k}\tilde{\z}_{k}}\R_{\tilde{\z}_{k}}^{-1} \R_{\tilde{\z}_{k}\y_{k}} \right)^{-1} 
			\R_{\y_{k}\tilde{\z}_{k}}
			{d\left(\R_{\tilde{\z}_{k}}^{-1}\right) 
			}
			\right)\notag \\
			& \quad -\tr \left(
			\left(\R_{\y_{k}} -  \R_{\y_{k}\tilde{\z}_{k}}\R_{\tilde{\z}_{k}}^{-1} \R_{\tilde{\z}_{k}\y_{k}} \right)^{-1} 
			\R_{\y_{k}\tilde{\z}_{k}}\R_{\tilde{\z}_{k}}^{-1} {d\left(\R_{\tilde{\z}_{k}\y_{k}}  \right)
			}
			\right) \label{diff}
			 .
		\end{align}
		By calculating each differential term in (\ref{diff}), we can obtain
		\begin{align}
			d\left(R_{k}^{\rm ub}\right)=  \tr 
			\left( \G_{k} d\left(\I_M \otimes d\left(\bar{\v}^H\right)  \right)
			\right) + \tr 
			\left( \Q_{k} d\left(\I_{M_e} \otimes d\left(\bar{\v}^H\right)  \right)
			\right).
		\end{align}
Therefore, the conjugate gradient of $R_{k}^{\rm ub}$ w.r.t. $\bar{\v}$ is calculated as
		\begin{align}
			\frac{\partial R_{k}^{\rm ub}}{\partial \bar{\v}^{*}}  = \sum_{m=1}^{M_e}\bar{\q}_{k,m} + \sum_{m=1}^{M}\bar{\g}_{k,m}.
		\end{align}
		This completes the proof.
%	\end{proof}

\end{appendices}

	\bibliographystyle{IEEEtran}
	\bibliography{IEEEabrv,Ref1}

\end{document}

%% file: macro_en.tex
\usepackage{cite}
\usepackage{graphicx}
\usepackage{times}
\usepackage{latexsym}

\usepackage[mathscr]{eucal}
\usepackage{array}
\usepackage{fancyhdr}
\usepackage{amsmath,amssymb}
\usepackage{bm} 
\usepackage{subfigure}
\usepackage{multirow}

\ifCLASSINFOpdf
\else
\fi

\newcommand{\defeq}{\stackrel{\Delta}{=}}

\newcommand{\vect}[1]{{\lowercase{\mathbf{#1}}}}
\newcommand{\mat}[1]{{\uppercase{\mathbf{#1}}}}

\newcommand{\tr}{{\rm{tr}}}
\newcommand{\diag}{{\rm{diag}}}
\renewcommand{\a}{\vect{a}} % _ accent
\renewcommand{\c}{\vect{c}} % , accent
\newcommand{\f}{\vect{f}}
\newcommand{\g}{\vect{g}}
\newcommand{\h}{\vect{h}}

\newcommand{\n}{\vect{n}}
\newcommand{\q}{\vect{q}}
\renewcommand{\r}{\vect{r}}

\renewcommand{\v}{\vect{v}} % v accent

\newcommand{\x}{\vect{x}}
\newcommand{\y}{\vect{y}}
\newcommand{\z}{\vect{z}}

\newcommand{\A}{\mat{A}}
\newcommand{\B}{\mat{B}}

\newcommand{\D}{\mat{D}}

\newcommand{\F}{\mat{F}}
\newcommand{\G}{\mat{G}}
\renewcommand{\H}{\mat{H}} % ''accent
\newcommand{\I}{\mat{I}}

\newcommand{\M}{\mat{M}}
\newcommand{\N}{\mat{N}}
\renewcommand{\P}{\mat{P}}
\newcommand{\Q}{\mat{Q}}
\newcommand{\R}{\mat{R}}

\newcommand{\X}{\mat{X}}

\newcommand{\Phim}{\hbox{\boldmath$\Phi$}}

\usepackage{xcolor}

\makeatother